\pdfoutput=1
\documentclass[final,finalrunningheads,a4paper]{llncs}
\usepackage{eucal}
\usepackage{ifthen}
\usepackage{xargs}
\usepackage{xcolor}
\usepackage{xspace}
\usepackage{xifthen}  
\usepackage{tikz}
\usetikzlibrary{cd}
\usepackage{mathbbol} 
\usepackage{nicefrac}
\usepackage{hyperref}
\usepackage{thm-restate}
\usepackage{amssymb}
\usepackage{commath}
\usepackage{etoolbox} 
\usepackage{algorithm}
\usepackage[noend]{algpseudocode}
\usepackage{cleveref}

\usepackage[final,nomargin,inline,index]{fixme} 
\fxusetheme{color}

\FXRegisterAuthor{bart}{anbart}{\color{magenta} {\underline{bart}}}
\FXRegisterAuthor{james}{anjames}{\color{blue} {\underline{james}}}
\FXRegisterAuthor{alberto}{analberto}{\color{ForestGreen} {\underline{alberto}}}

\hypersetup{
  breaklinks   = true,
  colorlinks   = true, 
  urlcolor     = blue, 
  linkcolor    = blue, 
  citecolor    = red   
}
\hypersetup{final} 

\usepackage[inline,shortlabels]{enumitem}
\newlist{inlinelist}{enumerate*}{1}
\setlist*[inlinelist,1]{%
  label=(\roman*),
}

\newcommand{\ifempty}[3]{%
  \ifthenelse{\isempty{#1}}{#2}{#3}%
}

\newcommand{\ifdots}[3]{%
  \ifthenelse{\equal{#1}{...}}{#2}{#3}%
}

\newcommand{\hidden}[1]{}

\newcommand{\mypar}[1]{\paragraph*{\textbf{#1}}}

\newcommand{\keyterm}[1]{\textbf{\emph{#1}}}%

\newcommand{\Real}[1]{\mathrm{Real}}

\newcommand{\eg}{e.g.\@\xspace}
\newcommand{\ie}{i.e.\@\xspace}
\newcommand{\wrt}{w.r.t.\@\xspace}

\newcommand{\emptyseq}{\varepsilon}

\renewcommand{\epsilon}{\varepsilon}

\newcommandx{\yaHelper}[2][1=\empty]{%
\ifthenelse{\equal{#1}{\empty}}%
  { \ensuremath{ \scriptstyle{ #2 } } } 
  { \raisebox{ #1 }[0pt][0pt]{ \ensuremath{ \scriptstyle{ #2 } } } }  
}

\newcommandx{\yrightarrow}[4][1=\empty, 2=\empty, 4=\empty, usedefault=@]{%
  \ifthenelse{\equal{#2}{\empty}}
  { \xrightarrow{ \protect{ \yaHelper[ #4 ]{ #3 } } } } 
  { \xrightarrow[ \protect{ \yaHelper[ #2 ]{ #1 } } ]{ \protect{ \yaHelper[ #4 ]{ #3 } } } } 
}

\newcommand{\myxrightarrow}[1]{\mathrel{\raisebox{-1pt}{$\yrightarrow{#1}[-2pt]$}}}

  {\medskip}

\newcommand{\BTC}{\textup{%
  \leavevmode
  \vtop{\offinterlineskip 
    \setbox0=\hbox{B}%
    \setbox2=\hbox to\wd0{\hfil\hskip-.03em
    \vrule height .3ex width .15ex\hskip .08em
    \vrule height .3ex width .15ex\hfil}
    \vbox{\copy2\box0}\box2}}\xspace}

\def\pmvColor{\color{ForestGreen}}

\newcommand{\pmvFmt}[1]{{\pmvColor{\sf{#1}}}} 

\newcommand{\pmv}[2][]{\pmvFmt{#2}_{\pmvColor{#1}}\xspace}

\newcommand{\pmvA}[1][]{\pmv[{#1}]{A}}

\newcommand{\pmvM}[1][]{\pmv[{#1}]{M}} 

\newcommand{\supply}[1]{\mathit{sply}_{#1}}

\newcommand{\res}[1][]{\mathit{res}_{#1}}

\newcommand{\resDeposit}[1][]{{\color{gray}\textsf{DF}}r_{#1}}
\newcommand{\resSwap}[1][]{{\color{gray}\textsf{SF}}r_{#1}}

\newcommand{\gain}[3][]{\mathit{G}_{#2}(\ifempty{#1}{}{{#1},}{#3})}
\newcommand{\mev}[2]{\mathit{MEV}({#1},{#2})}

\newcommand{\swapTx}[4][]{\txX^{#1}({#2},{#3}\ifempty{#4}{}{,{#4}})}
\newcommand{\swapVal}[5][]{\valW[{#1}]^{#2}\ifempty{#3}{}{({#3},{#4}\ifempty{#5}{}{,{#5}})}}

\newcommand{\frontSwapTx}[4][]{{\color{gray}\textsf{SF}}\txX^{#1}({#2},{#3}\ifempty{#4}{}{,{#4}})}
\newcommand{\frontSwapVal}[5][]{{\color{gray}\textsf{SF}}\valW[{#1}]^{#2}({#3},{#4}\ifempty{#5}{}{,{#5}})}
\newcommand{\frontDepTx}[4][]{{{\color{gray}\textsf{DF}}}\txX^{#1}({#2},{#3}\ifempty{#4}{}{,{#4}})}
\newcommand{\frontDepVal}[5][]{{\color{gray}\textsf{DF}}\valW[{#1}]^{#2}({#3},{#4}\ifempty{#5}{}{,{#5}})}

\newcommand{\swapDir}[1][]{\mathit{d}_{#1}}

\def\txColor{\color{MidnightBlue}}

\newcommand{\txFmt}[1]{{\txColor{\sf #1}}}

\newcommand{\tx}[2][]{\txFmt{#2}_{\txColor{#1}}}
\newcommand{\txT}[1][]{\tx[#1]{T}}
\newcommand{\txX}[1][]{\tx[#1]{X}}

\newcommand{\txTi}[1][]{\txFmt{T'_{\txColor{{#1}}}}}

\newcommand{\TxT}[1][]{{\txColor{\tx[#1]{\mathcal{X}}}}}
\newcommand{\TxTi}[1][]{{\txColor{\tx[#1]{\mathcal{X}'}}}}

\newcommand{\bcB}[1][]{{\txColor{\lambda_{#1}}}}
\newcommand{\bcBi}[1][]{{\txColor{\lambda'_{#1}}}}

\newcommand{\mmid}{\,\|\,}

\newcommand{\irule}[2]{\dfrac{#1}{#2}}

\newcommand{\nrule}[1]{{\scriptsize \textsc{#1}}}

\newcommand{\dom}[1]{\operatorname{dom} {#1}}

\newcommand{\RNN}{\mathbb{R}_0^{+}}

\newcommand{\bind}[2]{\nicefrac{#2}{#1}}
\newcommand{\setenum}[1]{\{#1\}}

\newcommand{\emptymset}{[]}
\newcommand{\msetenum}[1]{\lbrack{#1}\rbrack}

\newcommand{\msetcup}{+}

\definecolor{LightGrey}{rgb}{0.95,0.95,0.95}
\definecolor{keyword}{HTML}{7F0055}

\definecolor{MidnightBlue}{RGB}{0,103,149}
\definecolor{ForestGreen}{RGB}{34,139,34}

\def\tokColor{\color{magenta}}
\newcommand{\tokFmt}[1]{{\tokColor{#1}}}
\newcommand{\ETH}{\ensuremath{\tokT[0]}\xspace}
\newcommand{\WETH}{\ensuremath{\tokT[1]}\xspace}
\newcommand{\tokT}[1][]{\tokFmt{\tau_{#1}}}
\newcommand{\tokTi}[1][]{\tokFmt{\tau'_{#1}}}

\newcommand{\TokU}[1][]{\tokFmt{\mathbb{T}_{#1}}} 

\newcommand{\ammDepositOp}{{\txColor{\sf dep}}}
\newcommand{\ammRedeemOp}{{\txColor{\sf rdm}}}
\newcommand{\ammSwapOp}{{\txColor{\sf swap}}}
\newcommand{\ammSwapLOp}{{\txColor{\sf swap}^0}}
\newcommand{\ammSwapROp}{{\txColor{\sf swap}^1}}
\newcommand{\ammSwapParOp}[1]{{\txColor{\sf swap}^{#1}}}

\newcommand{\actAmmDeposit}[5]{\ifempty{#1}{}{{#1}:}\ammDepositOp({#2}:{#3},{#4}:{#5})}
\newcommand{\actAmmSwapL}[5]{\ifempty{#1}{}{{#1}:}\ammSwapLOp({#2}:{#3},{#4}:{#5})}
\newcommand{\actAmmSwapR}[5]{\ifempty{#1}{}{{#1}:}\ammSwapROp({#2}:{#3},{#4}:{#5})}
\newcommand{\actAmmSwap}[6]{\ifempty{#1}{}{{#1}:}\ammSwapParOp{#2}({#3}:{#4},{#5}:{#6})}
\newcommand{\actAmmRedeem}[2]{\ifempty{#1}{}{{#1}:}\ammRedeemOp({#2})}

\newcommand{\tokBal}[1][]{\sigma_{#1}} 
\newcommand{\tokBali}[1][]{\sigma'_{#1}} 
\newcommand{\tokBalii}[1][]{\sigma''_{#1}} 

\newcommand{\wal}[2]{{#1}[{#2}]}
\newcommand{\walA}[2][]{\wal{\pmvA[#1]}{#2}}

\newcommand{\walM}[2][]{\wal{\pmvM[#1]}{#2}}

\newcommand{\txType}[1]{\mathit{type}\ifempty{#1}{}{({#1})}}

\newcommand{\txUsr}[1]{\mathit{usr}({#1})}

\newlength\replength
\newcommand\repfrac{.1}

\newcommand\rulewidth{.6pt}
\newcommand\tdashfill[1][\repfrac]{\cleaders\hbox to \replength{%
  \smash{\rule[\arraystretch\ht\strutbox]{\repfrac\replength}{\rulewidth}}}\hfill}

\newcommand\tdotfill[1][\repfrac]{\cleaders\hbox to \replength{%
  \smash{\raisebox{\arraystretch\dimexpr\ht\strutbox-.1ex\relax}{.}}}\hfill}

\newcommand{\val}[2][]{#2_{#1}} 
\newcommand{\valV}[1][]{\val[#1]{v}}
\newcommand{\valVi}[1][]{\val[#1]{v'}}

\newcommand{\valW}[1][]{\val[#1]{w}}

\newcommand{\confG}[1][]{\Gamma_{#1}}
\newcommand{\confGi}[1][]{\Gamma'_{#1}}
\newcommand{\confGii}[1][]{\Gamma''_{#1}}
\newcommand{\confGiii}[1][]{\Gamma'''_{#1}}
\newcommand{\confGiiii}[1][]{\Gamma''''_{#1}}

\newcommand{\ammR}[1][]{r_{#1}} 
\newcommand{\ammRi}[1][]{r'_{#1}}

\newcommand{\exchOInit}{{\it P}}
\newcommand{\exchO}[2][]{{\it P}_{#1}({#2})}

\makeatletter
\renewcommand\paragraph{\@startsection{paragraph}{4}{\z@}%
  {2.25ex \@plus 1ex \@minus .2ex}%
  {-0.75em}%
  {\normalfont\normalsize\bfseries}}
\makeatother

\newtoggle{arxiv}
\toggletrue{arxiv}
\newtoggle{anonymous}

\begin{document}
\title{Maximizing Extractable Value from \\Automated Market Makers}

\iftoggle{anonymous}{}{
\author{Massimo Bartoletti\inst{1},
James Hsin-yu Chiang\inst{2},
Alberto {Lluch Lafuente}\inst{2}}
}

\institute{
Universit\`a degli Studi di Cagliari, Cagliari, Italy
\and
Technical University of Denmark, DTU Compute, Copenhagen, Denmark
}

\maketitle

\begin{abstract}
  Automated Market Makers (AMMs) are decentralized applications
  that allow users to exchange crypto-tokens
  without the need for a matching exchange order.
  AMMs are one of the most successful DeFi use cases:
  indeed, major AMM platforms 
  process a daily volume of transactions worth USD billions.
  Despite their popularity, AMMs are well-known to suffer from
  transaction-ordering issues:
  adversaries can influence the ordering of user transactions,
  and possibly front-run them with their own,
  to extract value from AMMs, to the detriment of users.
  We devise an effective procedure to construct a strategy 
  through which an adversary can \emph{maximize} the value extracted 
  from user transactions.
\end{abstract}


\newcommand{\uniswapsupply}{\$5.45B\xspace}
\newcommand{\curvesupply}{\$7.5B\xspace}
\newcommand{\uniswapvolumedaily}{\$860M\xspace}
\newcommand{\curvevolumedaily}{\$290M\xspace}
\newcommand{\statsdate}{May 2021\xspace}

\section{Introduction} \label{sec:intro}

Decentralized finance (DeFi) is emerging as
an alternative to traditional finance,
boosted by blockchains, crypto-tokens and smart contracts~\cite{Werner21sok}.
\emph{Automated Market Makers (AMMs)}
--- one of the main DeFi applications ---
allow users to exchange crypto-tokens 
without the need to find another party wanting to participate in the exchange.
%
Major AMM platforms 
like \eg Uniswap, Curve Finance, and SushiSwap, hold dozens of billions
of USD and process hundreds of millions
worth of transactions daily
\cite{uniswapstats,curvestats,sushistats}.

AMMs are sensitive to \emph{transaction-ordering attacks},
where adversaries who can influence the
ordering of transactions in the blockchain
exploit this power to \emph{extract value} from user transactions
\cite{Daian19flash,Eskandari19sok,Qin21quantifying,Zhou21high}.
We illustrate this kind of attacks through a minimal example.
Assume a Uniswap-like AMM
holding 100 units of a crypto-token \ETH and 100 units of another token~\WETH,
and assume that both tokens have the same price in the reference currency
(say, USD 1,000).
Now, suppose that user $\pmvA$ wants to swap
20 units of \ETH in her wallet for at least 15 units of \WETH.
This requires to append to the blockchain a transaction of the form
\mbox{$\actAmmSwapL{\pmvA}{20}{\ETH}{15}{\WETH}$},
where the prefix $\pmvA$ indicates the wallet involved in the transaction,
$\ammSwapOp$ is the called AMM function, and
the superscript $0$ indicates the swap direction,
\ie deposit \mbox{$20:\ETH$} to receive back at least \mbox{$15:\WETH$}
(a superscript $1$ would indicate the opposite direction).
In a \emph{constant-product} AMM platform like Uniswap,
the actual amount of \WETH transferred to $\pmvA$ must be
such that the product between the AMM reserves remains constant
before and after a swap.

Now, suppose that an adversary $\pmvM$ (possibly a miner)
observes $\pmvA$'s transaction in the txpool,
and appends to the blockchain the following \emph{sandwich}:
\[
  \actAmmSwapL{\pmvM\!}{5.9}{\ETH}{5.5}{\WETH}
  \;\;
  \actAmmSwapL{\pmvA\!}{20}{\ETH}{15}{\WETH}
  \;\;
  \actAmmSwapR{\pmvM\!}{25.9}{\ETH}{20.6}{\WETH}
\]
where the last transaction is in the opposite direction, \ie
$\pmvM$ sends \mbox{$20.6:\WETH$}
to receive at least \mbox{$25.9:\ETH$}.
As a result, $\pmvA$ only yields the \emph{minimum} amount of
$15:\tokT[1]$ in return for $20:\tokT[0]$.
This implies that USD 5,000 have been gained by $\pmvM$
and lost by $\pmvA$.
This has been called \emph{Miner Extractable Value} (MEV)~\cite{Daian19flash}.

Recent works study this and other kinds of attacks
to AMMs~\cite{Daian19flash,Qin21quantifying,Zhou21discovery,Zhou21high}:
however, all these approaches are preeminently \emph{empirical},
as they focus on the definition of heuristics to extract value from AMMs,
and on their evaluation in the wild.
To the best of our knowledge, a general solution to obtain \emph{optimal}
MEV is still missing,
even in the special case of constant-product AMMs.

To exemplify a case where prior approaches fail to extract optimal MEV,
consider the following set of user transactions,
containing a swap of \ETH for \WETH,
a deposit of units of \ETH and \WETH,
and a redeem of units of minted (liquidity) tokens:
\[
  \{\quad
  \actAmmSwapL{\pmvA}{40}{\ETH}{35}{\WETH},\;
  \actAmmDeposit{\pmvA}{30}{\ETH}{40}{\WETH},\;
  \actAmmRedeem{\pmvA}{10:(\ETH,\WETH)}
\quad\}
\]

Here, both the $\ammSwapOp$ and the $\ammDepositOp$ transactions
would be rejected.
For instance, the constant-product invariant dictates that
\mbox{$40:\ETH$} sent by the user swap in the initial AMM state
$(100:\tokT[0],100:\tokT[1])$ will return exactly \mbox{$28.6:\WETH$};
since the $\ammSwapOp$ transaction requires \mbox{$35:\WETH$},
it would be discarded.
The known heuristics here fail to extract any value.
Even considering only the $\ammSwapOp$,
the sandwich would not be profitable for $\pmvM$,
since it requires the \emph{same} direction
for $\pmvM$'s and $\pmvA$'s $\ammSwapOp$
(offer \ETH to obtain \WETH),
making $\pmvA$'s $\ammSwapOp$ not enabled.
Further, the known heuristics only operate on $\ammSwapOp$ actions,
neglecting user deposits and redeems.
%
This paper proposes a layered construction to extract the \emph{maximum value}
from all user transactions,
through a multi-layer sandwich that we call \emph{Dagwood sandwich}.
In our example, $\pmvM$'s strategy would be to fire the following
three-layer sandwich:
\begin{align*}
  & \actAmmSwapR{\pmvM}{11}{\ETH}{13}{\WETH}
  \quad
  \actAmmSwapL{\pmvA\,}{40}{\ETH}{35}{\WETH}
  \\
  & \actAmmSwapR{\pmvM}{42}{\ETH}{38}{\WETH}
    \quad
    \actAmmDeposit{\pmvA}{30}{\ETH}{40}{\WETH }
  \\
  & \actAmmSwapL{\pmvM}{18}{\ETH}{21}{\WETH}
\end{align*}

The first transaction is a $\ammSwapOp$ in the opposite direction
(\ie, pay \WETH to get~\ETH) \wrt the subsequent user $\ammSwapOp$,
unlike in the classical sandwich heuristic.
$\pmvM$'s second $\ammSwapOp$ enables $\pmvA$'s deposit;
the final $\ammSwapOp$ is an arbitrage move \cite{ammTheory}.
The user redeem is dropped, since it would negatively contribute to
$\pmvM$'s profit.
By firing the transaction sequence above,
$\pmvM$ can extract approx.\ USD 5,700 from $\pmvA$,
improving over $\ammSwapOp$-only attacks, that would only extract USD 5,000.

\mypar{Contributions}

To the best of our knowledge, this work is the first to formalise
the \emph{MEV game} for AMMs (\Cref{sec:mining-game}),
and the first to effectively construct optimal solutions
which attack all types of transactions
supported by constant-product AMMs (\Cref{sec:mining-solution}).
We discuss in~\Cref{sec:conclusions} the applicability of our
technique in the wild.
\iftoggle{arxiv}
         {The proofs of our statements are in \Cref{sec:proofs}.}
         {The proofs of our statements are in~\cite{BCL21arxiv}.}


\section{Automated Market Makers}
\label{sec:model}


We assume a set  $\TokU[0]$ of \keyterm{atomic token types}
(ranged over by $\tokT,\tokTi,\ldots$),
representing native cryptocurrencies
and application-specific tokens.
We denote by $\TokU[1] = \TokU[0] \times \TokU[0]$
the set of \keyterm{minted token types},
representing shares in AMMs.
In our model, tokens are \emph{fungible},
\ie individual units of the same type are interchangeable.
In particular, amounts of tokens of the same type
can be split into smaller parts,
and two amounts of tokens of the same type can be joined.
We use $\valV, \valVi, \ammR, \ammRi$ to range over
nonnegative real numbers ($\RNN$), and
we write \mbox{$\ammR:\tokT$} to denote $\ammR$ units
of token type $\tokT \in \TokU = \TokU[0] \cup \TokU[1]$.

We model the \keyterm{wallet} of a user $\pmvA$ as a term
$\wal{\pmvA}{\tokBal}$, where the partial map
$\tokBal \in \TokU \rightharpoonup \RNN$
represents $\pmvA$'s token holdings,
and write $\wal{\pmvA}{\_}$ if the wallet balance is clear from context.
We denote with $\dom(\tokBal)$ the domain of $\tokBal$.
An \keyterm{AMM}
is a pair of the form \mbox{$(\ammR[0]:\tokT[0],\ammR[1]:\tokT[1])$},
representing the fact that the AMM is holding
$\ammR[0]$ units of $\tokT[0]$ and $\ammR[1]$ units of $\tokT[1]$.
We denote by $\res[{\tokT[0],\tokT[1]}](\confG)$
the \emph{reserves} of $\tokT[0]$ and $\tokT[1]$ in $\confG$,
\ie $\res[{\tokT[0],\tokT[1]}](\confG) = (\ammR[0],\ammR[1])$
if $(\ammR[0]:\tokT[0],\ammR[1]:\tokT[1])$ is in $\confG$.


A \keyterm{state} is a composition of wallets and AMMs,
represented as a term:
\[
\wal{\pmvA[1]}{\tokBal[1]} \mid \cdots \mid \wal{\pmvA[n]}{\tokBal[n]}
\mid
(\ammR[1]:\tokT[1],\ammRi[1]:\tokTi[1])
\mid \cdots \mid
(\ammR[k]:\tokT[k],\ammRi[k]:\tokTi[k])
\]
where:
\begin{inlinelist}
\item all $\pmvA[i]$ are distinct,
\item
  the token types in an AMM are \emph{distinct}, and
\item
  distinct AMMs cannot hold exactly the same token types.
\end{inlinelist}
Note that two AMMs can have a common token type $\tokT$,
as in \mbox{$(\ammR[1]:\tokT[1],\ammR:\tokT) \mid (\ammRi:\tokT,\ammR[2]:\tokT[2])$},
thus enabling indirect trades between token pairs
not directly provided by any AMM.
We use $\confG, \confGi, \ldots$ to range over states.
For a base term $Q$ (either wallet or AMM), we write $Q \in \confG$
when $\confG = Q \mid \confGi$, for some $\confGi$,
where we assume that two states are equivalent
when they contain the same base terms.


We define the \keyterm{supply} of a token type $\tokT$ in a state $\confG$
as the sum of the balances of $\tokT$ in all the wallets and the AMMs
occurring in $\confG$.
Formally:
\begin{equation*}
    \supply{\tokT}(\walA{\tokBal}) =
    \begin{cases}
      \tokBal(\tokT) & \text{if $\tokT \in \dom(\tokBal)$} \\
      0 &\text{otherwise }
    \end{cases}
          \quad
    \supply{\tokT}(\ammR[0]:\tokT[0],\ammR[1]:\tokT[1]) =
    \begin{cases}
      \ammR[i] & \text{if $\tokT = \tokT[i]$} \\
      0 & \text{otherwise}
    \end{cases}
\end{equation*}
and the supply of $\tokT$ in $\confG \mid \confGi$
is the summation $\supply{\tokT}(\confG) + \supply{\tokT}(\confGi)$.


We model the interaction between users and AMMs
as a transition system between states.
A transition $\confG \myxrightarrow{\txT} \confGi$
represents the evolution of the state $\confG$ into $\confGi$ upon the
execution of the \keyterm{transaction} $\txT$.
The possible transactions are:
\begin{itemize}

\item \mbox{$\actAmmDeposit{\pmvA}{\valV[0]}{\tokT[0]}{\valV[1]}{\tokT[1]}$},
  which allows $\pmvA$ to \keyterm{deposit}
  \mbox{$\valV[0]:\tokT[0]$} and \mbox{$\valV[1]:\tokT[1]$}
  to an AMM, receiving
  in return units of the minted token $(\tokT[0],\tokT[1])$.

\item \mbox{$\actAmmSwap{\pmvA}{d}{\valV[0]}{\tokT[0]}{\valV[1]}{\tokT[1]}$}
  with $d \in \setenum{0,1}$,
  which allows $\pmvA$ to \keyterm{swap} tokens,
  \ie transfer $\valV[d]:\tokT[d]$ to an AMM,
  and receive in return \emph{at least} $\valV[1-d]:\tokT[1-d]$.

\item \mbox{$\actAmmRedeem{\pmvA}{\valV:\tokT}$},
  which allows to $\pmvA$ \keyterm{redeem}
  $\valV$ units of minted token $\tokT = (\tokT[0],\tokT[1])$
  from an AMM, receiving in return units of the atomic tokens
  $\tokT[0]$ and $\tokT[1]$.

\end{itemize}


\noindent
We now formalise the one-step relation $\myxrightarrow{\txT}$
through rewriting rules, inspired by~\cite{ammTheory}.
We use the standard notation $\tokBal \setenum{\bind{x}{v}}$
to update a partial map $\tokBal$ at point $x$:
namely, $\tokBal \setenum{\bind{x}{v}}(x) = v$, while
$\tokBal \setenum{\bind{x}{v}}(y) = \tokBal(y)$ for $y \neq x$.
For a partial map $\tokBal \in \TokU \rightharpoonup \RNN$,
a token type $\tokT \in \TokU$ and a partial operation
$\circ \in \RNN \times \RNN \rightharpoonup \RNN$,
we define the partial map $\tokBal \circ \valV:\tokT$
(updating $\tokT$'s balance in $\tokBal$ by $\valV$) as follows:
\begin{equation*} \label{eq: balance update}
  \tokBal \circ \valV:\tokT = \begin{cases}
    \tokBal\setenum{\bind{\tokT}{\tokBal(\tokT) \;\circ\; \valV}}
    & \text{if $\tokT \in \dom{\tokBal}$ and $\tokBal(\tokT) \circ \valV \in \RNN$}
    \\
    \tokBal\setenum{\bind{\tokT}{\valV}}
    & \text{if $\tokT \not\in \dom{\tokBal}$}
  \end{cases}
\end{equation*}


\mypar{Deposit}

Any user can create an AMM for a token pair $(\tokT[0],\tokT[1])$,
provided that such an AMM is not already present in the state.
This is achieved by the transaction
$\actAmmDeposit{\pmvA}{\valV[0]}{\tokT[0]}{\valV[1]}{\tokT[1]}$,
through which $\pmvA$ transfers $\valV[0]:\tokT[0]$ and $\valV[1]:\tokT[1]$
to the new AMM.
In return, $\pmvA$ receives an amount of units
of a new token type $(\tokT[0],\tokT[1])$, which is minted by the AMM.
We formalise this behaviour by the rule:
\begin{small}
\[
\irule
{
    \tokBal(\tokT[i]) \geq \valV[i] > 0
    \;\; (i \in \setenum{0,1})
    \quad
    \tokT[0] \neq \tokT[1]
    \quad
    \tokT[0], \tokT[1] \in \TokU[0]
    \quad
    (\_:\tokT[0],\_:\tokT[1]),(\_:\tokT[1],\_:\tokT[0]) \not \in \confG
}{
  \begin{array}{l}
    \walA{\tokBal} \mid \confG
    \xrightarrow{\actAmmDeposit{\pmvA}{\valV[0]}{\tokT[0]}{\valV[1]}{\tokT[1]}}
    \walA{\tokBal - \valV[0]:\tokT[0] - \valV[1]:\tokT[1] + \valV[0]:(\tokT[0],\tokT[1])} \mid
    (\valV[0]:\tokT[0],\valV[1]:\tokT[1]) \mid
    \confG
  \end{array}
}
\nrule{[Dep0]}
\]
\end{small}


Once an AMM is created, any user can deposit tokens into it,
as long as doing so preserves the ratio of the token holdings in the AMM.
When a user deposits $\valV[0]:\tokT[0]$ and $\valV[1]:\tokT[1]$
to an existing AMM, it receives in return an amount of minted tokens
of type $(\tokT[0],\tokT[1])$.
This amount is the ratio between the deposited amount $\valV[0]$ and the
redeem rate of $(\tokT[0],\tokT[1])$ in the current state $\confG$.
This redeem rate is the ratio between the amount $\ammR[0]$ of $\tokT[0]$
stored in the AMM, and the total supply $\supply{(\tokT[0],\tokT[1])}(\confG)$
of the minted token in the state.
\begin{small}
\[
\irule{
    \tokBal(\tokT[i]) \geq \valV[i] > 0
    \;\; (i \in \setenum{0,1})
    \qquad
    \ammR[1]\valV[0] = \ammR[0]\valV[1]
    \qquad
    \valV = \frac{\valV[0]}{\ammR[0]} \cdot \supply{(\tokT[0],\tokT[1])}{(\confG)}
}
{
  \begin{array}{ll}
    \confG \; = \;
    & \walA{\tokBal}
      \; \mid \;
      (\ammR[0]:\tokT[0],\ammR[1]:\tokT[1])
      \; \mid \;
      \confGi
      \xrightarrow{\actAmmDeposit{\pmvA}{\valV[0]}{\tokT[0]}{\valV[1]}{\tokT[1]}}
    \\[4pt]
    & \walA{\tokBal - \valV[0]:\tokT[0] - \valV[1]:\tokT[1] + \valV:(\tokT[0],\tokT[1])}
      \; \mid \;
      (\ammR[0]+\valV[0]:\tokT[0], \ammR[1]+\valV[1]:\tokT[1])
      \; \mid \;
      \confGi
  \end{array}
}
\nrule{[Dep]}
\]
\end{small}
The premise $\ammR[1]\valV[0] = \ammR[0]\valV[1]$
ensures that the ratio between the reserves
of $\tokT[0]$ and $\tokT[1]$ in the AMM is preserved, \ie
\(
\nicefrac{\ammR[1]+\valV[1]}{\ammR[0]+\valV[0]} \; = \; \nicefrac{\ammR[1]}{\ammR[0]}
\).

\mypar{Swap}

Any user $\pmvA$ can swap units of $\tokT[0]$ in her wallet
for units of $\tokT[1]$ in an AMM
\mbox{$(\ammR[0]:\tokT[0],\ammR[1]:\tokT[1])$},
or \emph{vice versa} swap units of $\tokT[1]$ in the wallet
for units of $\tokT[0]$ in the AMM.
This is achieved by the transaction
\mbox{$\actAmmSwap{\pmvA}{d}{\valV[0]}{\tokT[0]}{\valV[1]}{\tokT[1]}$},
where $d \in \setenum{0,1}$ is the \emph{swap direction}.
If $d=0$ (``left'' swap),
then $\valV[0]$ is the amount of $\tokT[0]$ transferred from
$\pmvA$'s wallet to the AMM, while $\valV[1]$ is a \emph{lower bound}
on the amount of $\tokT[1]$ that $\pmvA$ will receive in return.
Conversely, if $d=1$ (``right'' swap),
then $\valV[1]$ is the amount of $\tokT[1]$ transferred from
$\pmvA$'s wallet, and $\valV[0]$ is a lower bound
on the received amount of $\tokT[0]$.
The actual amount $\valV$ of received units of $\tokT[1-d]$ must satisfy
the \keyterm{constant-product invariant}~\cite{rvammspec},
as in Uniswap~\cite{uniswapimpl}, SushiSwap~\cite{sushiswapimpl}
and other common AMMs implementations:
\[
\ammR[0] \cdot \ammR[1]
\; = \;
(\ammR[d] + \valV[d]) \cdot (\ammR[1-d] - \valV)
\]
Formally, for $d \in \setenum{0,1}$ we define:
\begin{small}
\[
\irule
{
  \begin{array}{l}
    \tokBal(\tokT[d]) \geq \valV[d] > 0
    \qquad
    \valV = \frac{\ammR[1-d] \cdot \valV[d]}{\ammR[d]+\valV[d]}
    \qquad
    0 < \valV[1-d] \leq \valV
  \end{array}
}
{\begin{array}{l}
   \walA{\tokBal}
   \mid
   (\ammR[0]:\tokT[0],\ammR[1]:\tokT[1])
   \mid
   \confG
   \xrightarrow{\actAmmSwap{\pmvA}{d}{\valV[0]}{\tokT[0]}{\valV[1]}{\tokT[1]}}
   \\[4pt]
   \walA{\tokBal - \valV[d]:\tokT[d] + \valV:\tokT[1-d]}
   \mid
   (\ammR[0]:\tokT[0],\ammR[1]:\tokT[1])
   + \valV[d]:\tokT[d] - \valV:\tokT[1-d]
   \mid
   \confG
   \hspace{-5pt}
 \end{array}
}
\nrule{[Swap]}
\]
\end{small}%

\noindent
where we define the update of the units of $\tokT$
in an AMM, for $\circ \in \setenum{+,-}$, as:
\begin{equation*}
  (\ammR[0]:\tokT[0],\ammR[1]:\tokT[1]) \circ \valV:\tokT = \begin{cases}
    (\ammR[0] \circ \valV:\tokT[0],\ammR[1]:\tokT[1])
    & \text{if $\tokT = \tokT[0]$ and $\ammR[0] \circ \valV \in \RNN$}
    \\
    (\ammR[0]:\tokT[0],\ammR[1] \circ \valV:\tokT[1])
    & \text{if $\tokT = \tokT[1]$ and $\ammR[1] \circ \valV \in \RNN$}
  \end{cases}
\end{equation*}

\mypar{Redeem}

Users can redeem units of a minted token $(\tokT[0],\tokT[1])$
for units of the underlying atomic tokens $\tokT[0]$ and $\tokT[1]$.
Each unit of $(\tokT[0],\tokT[1])$ can be redeemed for
equal fractions of $\tokT[0]$ and $\tokT[1]$ remaining in the AMM:
\begin{small}
\[
\irule{
  \begin{array}{l}
    \tokBal(\tokT[0],\tokT[1]) \geq \valV > 0
    \qquad
    \valV[0] = \valV \frac{\ammR[0]}{\supply{(\tokT[0],\tokT[1])}{(\confG)}}
    \qquad
    \valV[1] = \valV \frac{\ammR[1]}{\supply{(\tokT[0],\tokT[1])}{(\confG)}}
  \end{array}
}
{
  \begin{array}{ll}
    \confG \; = \;
    & \walA{\tokBal}
      \; \mid \;
      (\ammR[0]:\tokT[0],\ammR[1]:\tokT[1])
      \; \mid \;
      \confGi
      \xrightarrow{\actAmmRedeem{\pmvA}{\valV:(\tokT[0],\tokT[1])}}
    \\[4pt]
    & \walA{\tokBal + \valV[0]:\tokT[0] + \valV[1]:\tokT[1] - \valV:(\tokT[0],\tokT[1])}
      \; \mid \;
      (\ammR[0]-\valV[0]:\tokT[0],\ammR[1]-\valV[1]:\tokT[1])
      \; \mid \;
      \confGi
  \end{array}
}
\nrule{[Rdm]}
\]
\end{small}


\noindent
A key property of the transition system is \emph{determinism},
\ie if $\confG \myxrightarrow{\txT} \confGi$ and
$\confG \myxrightarrow{\txT} \confGii$, then the states
$\confGi$ and $\confGii$ are equivalent.
We denote with $\txType{\txT}$ the type of $\txT$
(\ie, $\ammDepositOp$, $\ammSwapOp$, $\ammRedeemOp$),
and with $\txUsr{\txT}$ the user issuing $\txT$.
For a sequence of transactions $\bcB = \txT[1] \cdots \txT[n]$, we write
$\confG \myxrightarrow{\bcB} \confGi$ whenever
there exist intermediate states $\confG[1], \ldots \confG[n-1]$ such that
\(
\confG
\myxrightarrow{\txT[1]}
\confG[1]
\myxrightarrow{\txT[2]}
\cdots
\myxrightarrow{\txT[n-1]}
\confG[n-1]
\myxrightarrow{\txT[n]}
\confGi
\).
When this happens, we say that $\bcB$ is \emph{enabled} in $\confG$,
or just $\confG \myxrightarrow{\bcB}$.
A state $\confG$ is \emph{reachable} if there exist
some $\confG[0]$ only containing wallets with atomic tokens
and some $\bcB$
such that $\confG[0] \myxrightarrow{\bcB} \confG$.

\section{The MEV game}
\label{sec:mining-game}

The model in the previous section defines how
the state of AMMs and wallets evolves upon a sequence of transactions,
but it does not specify how this sequence is formed.
We specify this as a single-player, single-round game
where the only player is an adversary $\pmvM$ who attempts to maximize its MEV.
Accordingly, we call this the \keyterm{MEV game}.
The \emph{initial state} of the game is given by a reachable state
$\confG$ (not including $\pmvM$'s wallet)
and by a finite multiset $\TxT$ of user transactions,
representing the pool of pending transactions
(also called \keyterm{txpool}).
The \emph{moves} of $\pmvM$ are pairs $(\tokBal,\bcB)$,
where $\tokBal$ is $\pmvM$'s initial balance,
and $\bcB$ is a sequence formed by (part of)
the transactions in $\TxT$, and by any number of $\pmvM$'s transactions.
We require that the sequence $\bcB$ in a move is enabled in $\confG$.
The MEV game assumes the following (see
\Cref{sec:conclusions} for a discussion thereof):
\begin{enumerate}
  \item Users balances in $\confG$ are sufficiently high to not interfere with
  the validity of any specific ordering of actions in $\TxT$.
  \item The balance $\tokBal$ of $\pmvM$ does not include minted tokens.
  \item The length of the sequence $\bcB$ is unbounded.
  \item Prices of atomic tokens are fixed throughout the game execution.
\end{enumerate}
Besides the above, some further assumptions are implied by our AMM model:
\begin{enumerate}
  \item[5.] AMMs only hold atomic tokens (this is a consequence of $\nrule{[Dep0]}$).
  \item[6.] Swap actions do not require fees (this is a consequence of $\nrule{[Swap]}$).
  \item[7.] There are no transaction fees.
  \item[8.] Interval constraints on received token amounts are modelled in swaps only.
\end{enumerate}
A \emph{solution} to the game is a move that maximizes $\pmvM$'s \emph{gain},
\ie the change in $\pmvM$'s net worth after performing
the sequence $\bcB$ from~$\confG$.
Intuitively, the net worth of a user is the overall \emph{value} of tokens
in her wallet.
To define it, we need to associate a \keyterm{price} to each token.
%
We assume that the prices of atomic tokens are given by an oracle
$\exchOInit \in \TokU[0] \rightarrow \RNN$: naturally, the MEV game solution
will need to be recomputed should the price of atomic tokens be updated.
The price $\exchO[\confG]{\tokT[0],\tokT[1]}$
of a minted token $(\tokT[0],\tokT[1])$ in a state $\confG$ is defined as follows:
\begin{equation}
  \label{eq:price:minted}
  \exchO[\confG]{\tokT[0],\tokT[1]}
  = \dfrac
  {\ammR[0] \cdot \exchO{\tokT[0]} + \ammR[1] \cdot \exchO{\tokT[1]}}
  {\supply{(\tokT[0],\tokT[1])}(\confG)}
  \quad \text{if $\res[{\tokT[0],\tokT[1]}](\confG) = (\ammR[0],\ammR[1])$}
\end{equation}

Minted tokens are priced such that
the net worth of a user is preserved when she deposits or redeems minted
tokens in her wallet.
We assume that the reserves in an AMM are
never reduced to zero in an execution, in order to preserve equality
of minted token prices between two states with equal reserves, thereby facilitating
proofs and analysis.
While our semantics of AMMs allows reserves to be emptied,
we note that this does not occur in
practice, as it would halt the operation of the respective AMM pair.
%
We define the \keyterm{net worth} of a user $\pmvA$ in a state $\confG$
such that $\walA{\tokBal} \in \confG$ as follows:
\begin{equation}
  \label{eq:net-worth:user}
  W_{\pmvA}(\confG) \; = \;
  \textstyle \sum_{\tokT \in \dom(\tokBal)}
  \tokBal(\tokT)\cdot \exchO[\confG]{\tokT}
\end{equation}
and we denote by $\gain[\confG]{\pmvA}{\bcB}$ the \keyterm{gain} of user $\pmvA$
upon performing a sequence of transactions $\bcB$ enabled in state $\confG$
(if $\bcB$ is not enabled, the gain is zero):
\begin{equation}
  \label{def:gain:A}
  \gain[\confG]{\pmvA}{\bcB}
  \; = \;
  W_{\pmvA}(\confGi) -W_{\pmvA}(\confG)
  \qquad \text{if $\confG \xrightarrow{\bcB} \confGi$}
\end{equation}
A \keyterm{rational player} is a player which,
for all initial states $(\confG,\TxT)$ of the game,
always chooses a move $(\tokBal,\bcB)$ that maximizes the function
$\gain[\walM{x} \mid \confG]{\pmvM}{\ y}$
on variables $x$ and $y$.
We define the \keyterm{miner extractable value}
in $(\confG,\TxT)$ as the gain obtained by a rational player
by applying such a solution $(\tokBal,\bcB)$, \ie:
\[
\mev{\confG}{\TxT} \; = \; \gain[\walM{\tokBal} \mid \confG]{\pmvM}{\bcB}
\]

\Cref{lma:wealth-conservation} states that
firing transactions preserves the \emph{global} net worth,
\ie the gains of some users are balanced by equal overall losses of other users.
\begin{restatable}{lemma}{wealthconservation}
  \label{lma:wealth-conservation}
  \(
  \textstyle
  \sum_{\pmvA} \gain[\confG]{\pmvA}{\txT} = 0
  \).
\end{restatable}
By using a simple inductive argument, we can extend
\Cref{lma:wealth-conservation} to sequences of transactions:
if $\confG \myxrightarrow{\bcB} \confGi$,
then the summation of the gains
$\gain[\confG]{\pmvA}{\bcB}$ over all users (including $\pmvM$) is 0.
Hence, the MEV game is zero-sum.
The following \namecref{lma:constant-wealth-step} ensures
that deposit and redeem actions do not directly affect the net worth of the user
who performs them.

\begin{restatable}{lemma}{constantwealthstep}
  \label{lma:constant-wealth-step}
  If $\txType{\txT} \in \{ \ammDepositOp, \ammRedeemOp \}$,
  then
  \(
  \textstyle
  \gain[\confG]{\txUsr{\txT}}{\txT} = 0
  \).
\end{restatable}

Finally, we note that prices of a minted token in two states are equal if the
reserve ratio in the two states are as well.
\begin{restatable}{lemma}{pricereserverelation}
  \label{lma:price-reserve-relation}
  Let $\confG \myxrightarrow{\bcB} \confGi$, and let
  $\res[{\tokT[0],\tokT[1]}](\confG) = (\ammR[0],\ammR[1])$,
  $\res[{\tokT[0],\tokT[1]}](\confGi) = (\ammRi[0],\ammRi[1])$.
  Then, $\exchO[\confG]{\tokT[0],\tokT[1]} = \exchO[\confGi]{\tokT[0],\tokT[1]}$
  if and only if
  $\ammR[0]/\ammR[1] = \ammRi[0]/\ammRi[1]$.
\end{restatable}
\noindent
\section{Solving the MEV game}
\label{sec:mining-solution}

By \Cref{lma:wealth-conservation},
a move which minimizes the gain of all users but $\pmvM$
must maximize $\pmvM$'s gain, and therefore is a solution to the MEV game.
More formally, we have:

\begin{corollary} 
  \label{lma:user-wealth-minimization}
  $\gain[\confG]{\pmvM}{\bcB}$ is maximized
  iff $\gain[\confG]{\pmvA}{\bcB}$ is minimized for all $\pmvA \neq \pmvM$.
\end{corollary}

The net worth $W_{\pmvA}$ of a user $\pmvA$
can be decomposed in two parts:
$W_{\pmvA}^0$, which accounts for the atomic tokens, 
and $W_{\pmvA}^1$, which accounts for the minted tokens: 
\begin{equation}
  \label{eq:net-worth-components}
  \textstyle
  W_{\pmvA}^0(\confG) = \sum_{\tokT \in \TokU[0]}
  \tokBal[\pmvA](\tokT)\cdot \exchO{\tokT}
  \qquad
  W_{\pmvA}^1(\confG) = \sum_{\tokT \in \TokU[1]}
  \tokBal[\pmvA](\tokT) \cdot \exchO[\confG]{\tokT}
\end{equation}

This provides $\pmvM$ with two levers to reduce the users' gain:
token balances, and the price of minted tokens.
To use the first lever, $\pmvM$ needs to exploit user actions in
the txpool $\TxT$ of the MEV game.
For the second lever,
since the prices of atomic tokens ($\tokT\in\TokU[0]$) are fixed,
$\pmvM$ can only influence the price of minted tokens
($\tokT\in\TokU[1]$).
This can be achieved by performing actions on the respective~AMMs.

In the rest of the~\namecref{sec:mining-solution} we
devise an \emph{optimal} strategy to exploit these two levers.
Intuitively, our strategy constructs a multi-layer \keyterm{Dagwood Sandwich}%
\footnote{We name it after Dagwood Bumstead, a comic strip character
who is often illustrated while producing enormous multi-layer sandwiches.},
containing an \keyterm{inner layer} for each exploitable user action in $\TxT$,
which $\pmvM$ \emph{front-runs} by a swap transaction
to enable it (if necessary),
and a \keyterm{final layer} of swaps by $\pmvM$ 
to minimize the prices of all minted tokens.

The construction of the final layer of the Dagwood sandwich
is shown in \S\ref{sec:price-minimization},
while the construction of the inner layers is presented in \S\ref{sec:game-soln}.

\subsection{Price minimization}
\label{sec:price-minimization}

\Cref{lma:price-minimum} below states that, in any state,
$\pmvM$ can minimize the price of a minted token by using a single swap,
at most.
In particular, this minimization can always be performed
in the final layer of the Dagwood sandwich.
%

\begin{restatable}{lemma}{priceminimum}
  \label{lma:price-minimum}
  There exists a function
  $\exchOInit^{\textit{min}}$ such that if
  \mbox{$\walM{\tokBal} \mid \confG \rightarrow^{*} \walM{\tokBali} \mid \confGi$}
  then:
  \begin{inlinelist}
  \item \mbox{$\exchO[{\confGi}]{\tokT[0],\tokT[1]} \geq \exchOInit^{\textit{min}}_{\confG}(\tokT[0],\tokT[1])$};
  \item
    there exist $\tokBalii$ and
    $\bcB$ consisting at most of a swap by $\pmvM$
    such that
    \mbox{$\walM{\tokBalii} \mid \confGi \myxrightarrow{\bcB} \walM{\_} \mid \confGii$}
    and
    $\exchO[{\confGii}]{\tokT[0],\tokT[1]} = \exchOInit^{\textit{min}}_{\confG}(\tokT[0],\tokT[1])$.
  \end{inlinelist}
\end{restatable}



In order to construct the $\ammSwapOp$ transaction
which minimizes the price of a minted token $(\tokT[0],\tokT[1])$
in $\confG$, we need some auxiliary definitions.
For each swap direction $d \in \setenum{0,1}$,
we define the \keyterm{canonical swap values} as:
\label{eq:price-minimization-values}
\begin{align*}
  \swapVal[d]{d}{\tokT[0],\tokT[1]}{\confG}{}
  = \sqrt{\tfrac{\exchO{\tokT[1-d]}}{\exchO{\tokT[d]}} \cdot \ammR[0] \cdot \ammR[1]} - \ammR[d]
    \qquad
  \swapVal[1-d]{d}{\tokT[0],\tokT[1]}{\confG}{}
  = \frac
    {\ammR[1-d] \cdot \swapVal[d]{d}{\tokT[0],\tokT[1]}{\confG}{}}
    {\ammR[d] + \swapVal[d]{d}{\tokT[0],\tokT[1]}{\confG}{}}
\end{align*}

Intuitively, $\swapVal[d]{d}{}{}{}$
is the amount of tokens \emph{deposited} in a swap of direction $d$:
it is defined such that, after the swap, the AMM
reaches an equilibrium, where the ratio of the AMM reserves is
equal to the (inverse) ratio of the token prices.
Instead, $\swapVal[1-d]{d}{}{}{}$ is
the amount of tokens \emph{received} after the swap,
\ie it is the unique value for which the swap invariant is satisfied.

If both
$\swapVal[0]{0}{\tokT[0],\tokT[1]}{\confG}{} \leq 0$
and
$\swapVal[1]{1}{\tokT[0],\tokT[1]}{\confG}{} \leq 0$,
then the price of the minted token $(\tokT[0],\tokT[1])$
is already minimized.
Otherwise, if
\(
\swapVal[d]{d}{\tokT[0],\tokT[1]}{\confG}{} > 0
\)
for some $d$ (and there may exist at most one $d$ for which this holds),
then we define the \keyterm{price minimization transaction}
$\swapTx[d]{\tokT[0],\tokT[1]}{\confG}{}$
as:
\begin{equation}
\label{eq:price-minimization-tx}
\actAmmSwap
{\pmvM}
{d}
{\,\swapVal[0]{d}{\tokT[0],\tokT[1]}{\confG}{}}{\tokT[0]}
{\;\swapVal[1]{d}{\tokT[0],\tokT[1]}{\confG}{}}{\tokT[1] \,}
\end{equation}

\Cref{thm:0-tx-solution:greedy} constructs the
final layer of the Dagwood sandwich.
We show that this layer is the solution of
the MEV game on an empty txpool.
This is because if $\pmvM$ cannot leverage user transactions,
the solution is just to minimize the price of all minted tokens.
The solution is obtained by sequencing price minimization transactions
on all AMMs.
Since the price of a minted token
is a function of the reserves of the corresponding AMM,
this can be done in any order.

\begin{restatable}{theorem}{zerotxsolution}
  \label{thm:0-tx-solution:greedy}
  Let \
  \mbox{$\confG = \mmid_{i \in I} (\ammR[i,0]:\tokT[i,0],\ammR[i,1]:\tokT[i,1]) \mid \confG[w]$},
  where $\confG[w]$ only contains wallets.
  For all $j \in I$ and $d \in \setenum{0,1}$, let
  $v_j^d = \swapVal[d]{d}{\tokT[j,0],\tokT[j,1]}{\confG}{}$, and let:
  \[
    \tokBal[j] = \begin{cases}
      v_j^d:\tokT[j,d] & \text{if $v_j^d > 0$} \\
      0 & \text{if $v_j^0,v_j^1 \leq 0$}
    \end{cases}
    \qquad
    \bcB[j] = \begin{cases}
      \swapTx[d]{\tokT[j,0],\tokT[j,1]}{\confG}{}
      & \text{if $v_j^d > 0$} \\
      \emptyseq & \text{if $v_j^0,v_j^1 \leq 0$}
    \end{cases}
  \]
  Then, $(\tokBal[1] \cdots  \tokBal[n],\bcB[1] \cdots \bcB[n])$
  is a solution to the game $(\confG,\TxT)$ for an empty $\TxT$.
\end{restatable}



\subsection{Constructing the inner layers}
\label{sec:game-soln}

Consider a solution $(\tokBal,\bcB)$ to the game
\mbox{$(\walA{\tokBal[\pmvA]}\mid\confG,\TxT)$}, and let:
\[
  \walM{\tokBal}\mid\walA{\tokBal[\pmvA]} \mid \confG
  \;\xrightarrow{\;\bcB\;}\;
  \walM{\tokBali}\mid\walA{\tokBali[\pmvA]}\mid\confGi
\]
By decomposing the net worth as in~\eqref{eq:net-worth-components},
we find that $\pmvA$'s gain for $\bcB$ is:
\begin{align}
  \nonumber
  & \gain[\walM{\tokBal}\mid\walA{\tokBal[\pmvA]}\mid\confG]{\pmvA}{\bcB}
    = W_{\pmvA}^{0}(\confGi) - W_{\pmvA}^{0}(\confG)+ W_{\pmvA}^{1}(\confGi) - W_{\pmvA}^{1}(\confG)
  \\
  \nonumber
  & = \sum_{\tokT \in \TokU[0]}
    \big(
    \tokBali[\pmvA](\tokT) - \tokBal[\pmvA](\tokT)
    \big)
    \cdot \exchO{\tokT}
    +
    \sum_{\tokT \in \TokU[1]}
    \big(
    \tokBali[\pmvA](\tokT) \cdot \exchOInit_{\confGi}(\tokT) - \tokBal[\pmvA](\tokT) \cdot \exchOInit_{\confG}(\tokT)
    \big)
    \intertext{Since $\bcB$ is a solution,
    by~\Cref{lma:price-minimum} we can replace $\exchO[\confGi]{\tokT}$
    with $\exchOInit^{\textit{min}}_{\confG}(\tokT)$:}
    \label{eq:usr-gain-solution}
  & = \sum_{\tokT \in \TokU[0]}
    \big(
    \tokBali[\pmvA](\tokT) - \tokBal[\pmvA](\tokT)
    \big)
    \cdot \exchO{\tokT}
    +
    \sum_{\tokT \in \TokU[1]}
    \big(
    \tokBali[\pmvA](\tokT) \cdot \exchOInit_{\confG}^{\textit{min}}(\tokT) - \tokBal[\pmvA](\tokT) \cdot \exchOInit_{\confG}(\tokT)
    \big)
\end{align}

Note that all token prices in \eqref{eq:usr-gain-solution}
are already defined in state $\confG$.
Thus, $\pmvA$'s gain can be minimized by considering only the effect
on the token balance $\tokBali[\pmvA]$,
which we can rewrite as
$\tokBal[\pmvA] + \Delta_0  + \Delta_1 + \cdots$
where $\Delta_i$ is the effect on user $\pmvA$'s balance induced
by the $i$'th transaction in $\bcB$: this transaction
is necessarily one initially authorized by $\pmvA$.
We will show that $\Delta_i$ is \emph{fixed} for any user transaction
when executed in an inner solution layer:
the position of an inner layer in solution $\bcB$
does not affect its optimality.



The following~\namecref{thm:n-tx-soln} states that
solutions to the MEV game can be constructed incrementally,
by layering the local solutions for each individual transaction
in the txpool.
Intuitively, we choose a transaction $\txT$ from $\TxT$,
we solve the game for $(\confG,[\txT])$,
we compute the state $\confGi$ obtained by executing this solution,
and we inductively solve the game in the $(\confGi,\TxTi)$,
where $\TxTi$ is $\TxT$ minus $\txT$.

\begin{restatable}{theorem}{ntxsolution}
  \label{thm:n-tx-soln}
  With respect to the MEV game in $(\confG,\TxT)$:
  \begin{enumerate}

  \item If $\TxT$ is empty,
    the solution is the final layer constructed for $(\confG,\emptymset)$ in
    \S\ref{sec:price-minimization}.

  \item Otherwise, if $\TxT = \msetenum{\txT} \msetcup \TxTi$,
    let $(\tokBal,\bcB)$ be the inner layer constructed for $(\confG,\msetenum{\txT})$,
    let \mbox{$\walM{\tokBal}\mid\confG \myxrightarrow{\bcB} \walM{\_}\mid\confGi$},
    and let $(\tokBali,\bcBi)$ be the solution for $(\confGi,\TxTi)$.
    Then, the solution to $(\confG,\TxT)$ is $(\tokBal+\tokBali,\bcB\bcBi)$.
  \end{enumerate}
\end{restatable}

We now describe how to define the inner layers of the Dagwood sandwich,
\ie the base case of the inductive construction given by~\Cref{thm:n-tx-soln}.
Each inner layer includes a user transaction from the txpool,
possibly front-run by $\pmvM$
such that executing the layer leads the user's net worth to a local minimum.
We define below the construction of these inner layers for each transaction type.

\mypar{Swap inner layer}

Swap actions only affect the balance of \emph{atomic} tokens.
To minimize the gain of $\pmvA$ after a swap, $\pmvM$ must make
$\pmvA$ receive exactly the \emph{minimum} amount of requested tokens.
The effect of the swap on $\pmvA$'s \emph{atomic net worth} is:
\begin{equation*}
  W_{\pmvA}^{0}(\confGi) - W_{\pmvA}^{0}(\confG)
  = - \valV[d] \cdot \exchO{\tokT[d]}
  + \valV[1-d]\cdot \exchO{\tokT[1-d]}
  \qquad
  \text{if }
  \confG \xrightarrow{\actAmmSwap{\pmvA}{d}{\valV[0]}{\tokT[0]}{\valV[1]}{\tokT[1]}} \confGi
\end{equation*}

If the change in $\pmvA$'s atomic net worth is negative,
$\pmvA$'s transaction is included in the layer.
Although this transaction minimizes $\pmvA$'s atomic net worth,
it simultaneously affects the price of the minted token $(\tokT[0],\tokT[1])$.
This is not an issue, since the final layer of the Dagwood sandwich
minimizes the prices of \emph{all} minted tokens.
Thus, the change of minted token prices due to the swap inner layer
will \emph{not} affect the user gain in the full Dagwood sandwich,
as evident from \eqref{eq:usr-gain-solution}.
Note that the amount of tokens exchanged in a swap
is chosen by the user, so the actual position of the layer
in the Dagwood sandwich is immaterial (\Cref{thm:n-tx-soln}).

We now define the transaction used by $\pmvM$ to front-run $\pmvA$'s swap,
ensuring that $\pmvA$ receives the least amount of tokens from the swap.
For \mbox{$\confG = (\ammR[0]:\tokT[0],\ammR[1]:\tokT[1]) \mid \cdots$} and
\mbox{$\txT = \actAmmSwap{\pmvA}{\swapDir[\pmvA]}{\valV[0]}{\tokT[0]}{\valV[1]}{\tokT[1]}$},
let the \keyterm{swap front-run reserves} be:
\begin{equation*}
  \begin{split}
    \resSwap[{\swapDir[\pmvA]}](\tokT[0],\tokT[1],\confG,\txT)
    & = \frac{\abs{\sqrt{\valV[0]^2\cdot\valV[1]^2 + 4\cdot\valV[0]\cdot\valV[1]\cdot\ammR[0]\cdot\ammR[1]}\;}-\valV[0]\cdot\valV[1]}{2\cdot\valV[{{1-\swapDir[\pmvA]}}]}
    \\
    \resSwap[{1-\swapDir[\pmvA]}](\tokT[0],\tokT[1],\confG,\txT)
    & = \frac{\ammR[0]\cdot\ammR[1]}{\resSwap[{\swapDir[\pmvA]}](\tokT[0],\tokT[1],\confG,\txT)}
  \end{split}
\end{equation*}

\noindent
These values define the reserves of $(\tokT[0],\tokT[1])$
in the state $\confGi$ reached from $\walM{\tokBal}\mid\confG$
with $\pmvM$'s transaction.
Intuitively, if the swap front-run reserves do \emph{not} coincide
with the reserves $\ammR[0]$, $\ammR[1]$ in $\confG$,
then $\pmvM$'s transaction is needed to enable $\pmvA$'s swap.
We define the \keyterm{swap front-run direction} $\swapDir[\pmvM]$ as:
\begin{equation*}
  \swapDir[\pmvM] = \begin{cases}
    \swapDir[\pmvA]
    & \text{if } \resSwap[{\swapDir[\pmvA]}](\tokT[0],\tokT[1],\confG,\txT) > \ammR[{\swapDir[\pmvA]}]
    \\
    1-\swapDir[\pmvA]
    & \text{if } \resSwap[{1-\swapDir[\pmvA]}](\tokT[0],\tokT[1],\confG,\txT) > \ammR[{1-\swapDir[\pmvA]}]
  \end{cases}
\end{equation*}

\noindent
We define the \keyterm{swap front-run values}
(\ie, the parameters of $\pmvM$'s swap) as:
\begin{equation}
\label{eq:canon-sf-values}
\begin{split}
  \frontSwapVal[{{\swapDir[\pmvM]}}]{}{\tokT[0],\tokT[1],\confG}{\txT}{}
  & = \begin{cases} \resSwap[{\swapDir[\pmvA]}](\tokT[0],\tokT[1],\confG,\txT)-\ammR[{\swapDir[\pmvA]}] & \text{if } \swapDir[\pmvM] = \swapDir[\pmvA] \\
    \ammR[{\swapDir[\pmvA]}] - \resSwap[{\swapDir[\pmvA]}](\tokT[0],\tokT[1],\confG,\txT) &  \text{if } \swapDir[\pmvM] = 1-\swapDir[\pmvA] \end{cases}
  \\
  \frontSwapVal[{{1-\swapDir[\pmvM]}}]{}{\tokT[0],\tokT[1],\confG}{\txT}{}
  & = \begin{cases} \ammR[{1-\swapDir[\pmvM]}]-\resSwap[{1-\swapDir[\pmvM]}](\tokT[0],\tokT[1],\confG,\txT) & \text{if } \swapDir[\pmvM] = \swapDir[\pmvA] \\
    \resSwap[{1-\swapDir[\pmvM]}](\tokT[0],\tokT[1],\confG,\txT)-\ammR[{1-\swapDir[\pmvM]}] &  \text{if } \swapDir[\pmvM] = 1-\swapDir[\pmvM] \end{cases}
\end{split}
\end{equation}

\noindent
We combine these values to craft the \keyterm{swap front-run transaction}:
\begin{align*}
  \frontSwapTx{\tokT[0],\tokT[1],\confG}{\txT}{}
  =
  \actAmmSwap
  {\pmvM}
  {\swapDir[\pmvM]}
  {\frontSwapVal[0]{}{\tokT[0],\tokT[1],\confG}{\txT}{}}{\tokT[0]}
  {\frontSwapVal[1]{}{\tokT[0],\tokT[1],\confG}{\txT}{}}{\tokT[1]}
\end{align*}

The inner layer
is included in the Dagwood sandwich
if it reduces $\pmvA$'s net worth, \ie
if $-\valV[\swapDir]\cdot  \exchO{\tokT[\swapDir]}
+ \valV[1-\swapDir]\cdot \exchO{\tokT[1-\swapDir]} < 0$.
The swap front-run transaction is omitted if the reserves in $\confG$
coincide with the swap front-run reserves.
The balance of $\pmvM$ in the (local) game solution is
$\frontSwapVal[{{\swapDir[\pmvM]}}]{}{\tokT[0],\tokT[1],\confG}{\txT}{}:\tokT[{{\swapDir[\pmvM]}}]$.
Note that, the amount of tokens exchanged by the swapping user in \eqref{eq:usr-gain-solution} is fixed by
$(-\valV[\swapDir], +\valV[1-\swapDir])$, and the effect of a swap inner layer does not depend
on its position along the Dagwood sandwich (\Cref{thm:n-tx-soln}).


\begin{example}
  \label{ex:1txpool-game-sf}
  We recast our first example in \S\ref{sec:intro}
  as a MEV game, assuming a txpool
  \mbox{$\TxT = \{ \actAmmSwapL{\pmvA}{40}{\tokT[0]}{35}{\tokT[1]} \}$}.
  The initial state is
  \mbox{$\confG = (100:\tokT[0],100:\tokT[1]) \mid\confG[w]$},
  where $\confG[w]$ is made of user wallets,
  among which $\walA{40:\tokT[0]}$,
  and $\exchO{\tokT[0]} = \exchO{\tokT[1]} = 1,000$.
  We construct the Dagwood sandwich.
  Since $\pmvA$'s swap yields a reduction in $\pmvA$'s atomic net worth,
  $35\cdot\exchO{\tokT[1]}-40\cdot\exchO{\tokT[0]} = -5,000$,
  then $\pmvA$'s transaction is included in the inner layer.
  To check if $\pmvA$'s swap must be front-run by $\pmvM$,
  we first compute the swap front-run reserves:
  \begin{align*}
    \resSwap[0](\tokT[0],\tokT[1],\txT,\confG)
    & = \frac{\sqrt{40^2\cdot35^2 +4\cdot40\cdot35 \cdot100^2}-40\cdot35}{2\cdot35} \approx 88.8
    \\
    \resSwap[1](\tokT[0],\tokT[1],\txT,\confG)
    & = \frac{100^2}{89} \approx 112.7
  \end{align*}
  Since these values
  differ from the reserves in the initial game state,
  $\pmvM$ must front-run $\pmvA$'s transaction.
  The direction $\swapDir[\pmvM]$ of $\pmvM$'s swap is $1$,
  as $\resSwap[1](\tokT[0],\tokT[1],\confG,\txT) > \ammR[1]$.
  The swap front-run values \eqref{eq:canon-sf-values} are given by:
  \begin{equation*}
    \frontSwapVal[0]{}{\tokT[0],\tokT[1],\confG}{\txT}{}
    = 100 - 88.8 \approx 11.2
    \quad
    \frontSwapVal[1]{}{\tokT[0],\tokT[1],\confG}{\txT}{}
    = 112.7 - 100 \approx 12.7
  \end{equation*}
  Therefore, the swap inner layer is made of two transactions:
  \begin{equation*}
    \actAmmSwap
    {\pmvM}
    {1}
    {11.2}{\tokT[0]}
    {12.7}{\tokT[1]}
    \quad
    \actAmmSwapL{\pmvA}{40}{\tokT[0]}{35}{\tokT[1]}
  \end{equation*}
  and $\pmvM$'s balance of the (local) game solution is
  \mbox{$12.7:\tokT[1]$}.
  To construct the final layer, we consider the state
  $\confGii = (128.8:\tokT[0],77.7:\tokT[1]) \mid \cdots$,
  shown in~\Cref{fig:swap-front-run}.

  \begin{figure}[b]
    \scalebox{1}{\parbox{\textwidth}{%
        \begin{align}
          \nonumber
          & \walM{35:\tokT[1]} \mid
            \confG =
            (100:\tokT[0],100:\tokT[1]) \mid \cdots
          \\
          \nonumber
          \xrightarrow{\frontSwapTx{\tokT[0],\tokT[1],\confG}{\txT}{}} \;
          & \walM{11.2:\tokT[0],22.3:\tokT[1]} \mid
            \confGi =
            (88.8:\tokT[0],112.7:\tokT[1]) \mid \cdots
          \\
          \nonumber
          \xrightarrow{\txT = \actAmmSwapL{\pmvA}{40}{\tokT[0]}{35}{\tokT[1]}} \;
          & \walM{11.2:\tokT[0],22.3:\tokT[1]} \mid
            \confGii =
            (128.8:\tokT[0],77.7:\tokT[1]) \mid \cdots
          \\
          \nonumber
          \xrightarrow{\swapTx{\tokT[0],\tokT[1]}{\confGii}{}} \;
          & \walM{40:\tokT[0],0:\tokT[1]} \mid
            \confGiii =
            (100:\tokT[0],100:\tokT[1]) \mid \cdots
        \end{align}}}
    \caption{A Dagwood sandwich exploiting a single user swap.}
    \label{fig:swap-front-run}
  \end{figure}

  \noindent
  In $\confGii$, the canonical swap values are given by:
  \begin{align*}
  \swapVal[0]{1}{\tokT[0],\tokT[1]}{\confGii}{}
  & = \frac
    {128.8 \cdot 22.3}
    {77.7 + 22.3} \approx 28.7
    \\
  \swapVal[1]{1}{\tokT[0],\tokT[1]}{\confGii}{}
    & = \sqrt{\tfrac{1}{1} \cdot 128.8 \cdot 77.7} - 77.7 \approx 22.3
  \end{align*}
  Since $\swapVal[1]{1}{\tokT[0],\tokT[1]}{\confGii}{} > 1$,
  the direction $d$ of the price minimization swap is $1$.
  Therefore, the final layer is made of a single swap
  on the pair $(\tokT[0],\tokT[1])$:
  \begin{equation*}
    \label{ex:1swap-0-tx-soln}
    \actAmmSwap
    {\pmvM}
    {1}
    {28.7}{\tokT[0]}
    {22.3}{\tokT[1]})
  \end{equation*}
  where $\pmvM$'s required balance is \mbox{$22.3:\tokT[1]$}.
  Summing up, the Dagwood sandwich is constructed by appending
  the final layer to the inner layer,
  and $\pmvM$'s required balance is
  \mbox{$\tokBal = 12.7:\tokT[1] + 22.3:\tokT[1] = 35:\tokT[1]$}.
  The MEV obtained by $\pmvM$ through the Dagwood sandwich is
  $(11.2 - 12.7)\cdot 1,000 + (28.7-22.3) \cdot 1,000 \approx 5,000$.
  \qed
\end{example}

\mypar{Deposit inner layer}

By \Cref{lma:constant-wealth-step}, deposits preserve the user's net worth.
Thus, executing
$\txT = \actAmmDeposit{\pmvA}{\valV[0]}{\tokT[0]}{\valV[1]}{\tokT[1]}$ in $\confG$
does not bring any gain to $\pmvA$:
\begin{equation}
  \label{eq:deposit-gain}
  \gain[{\confG}]
  {\pmvA}
  {\txT}
  = - \valV[0] \cdot \exchO{\tokT[0]} - \valV[1] \cdot \exchO{\tokT[1]}
  + \valV \cdot \exchO[\confG]{\tokT[0],\tokT[1]} = 0
\end{equation}
where $\valV$ is the amount of minted tokens $(\tokT[0],\tokT[1])$
given to $\pmvA$ upon the deposit.
By~\Cref{lma:price-minimum},
$\exchO[\confG]{\tokT[0],\tokT[1]} \geq \exchOInit^{\textit{min}}_{\confG}(\tokT[0],\tokT[1])$.
By using this inequality in \eqref{eq:deposit-gain}, we have:
\begin{align*}
  & -\valV[0] \cdot \exchO{\tokT[0]}
  -\valV[1] \cdot \exchO{\tokT[1]}
  + \valV \cdot \exchOInit^{\textit{min}}_{\confG}(\tokT[0],\tokT[1])
  \leq 0
  \\
  \iff
  & \, \valV \cdot \exchOInit^{\textit{min}}_{\confG}(\tokT[0],\tokT[1])
  \leq
  \valV[0] \cdot \exchO{\tokT[0]}
  +\valV[1] \cdot \exchO{\tokT[1]}
\end{align*}

By \eqref{eq:usr-gain-solution} it follows that
including $\txT$ in a game solution $\bcB$ reduces $\pmvA$'s net worth,
since the decrease of $\pmvA$'s net worth
in atomic tokens is \emph{not} always offset
by the increase of net worth in minted tokens.
Additionally, the minted token price $\exchO[\confG]{\tokT[0],\tokT[1]}$
in \eqref{eq:deposit-gain} when the user deposit occurs is determined by
deposit parameters $\valV[0]$, $\valV[1]$ alone:
let $\confG \rightarrow^{*} \confGi$ be such that the given user deposit
$\txT$
is \emph{enabled} in both $\confG$ and $\confGi$. By \nrule{[Dep]}, this implies
$\valV[0]/\valV[1] = \ammR[0]/\ammR[1] = \ammRi[0]/\ammRi[1]$ where
$(\ammR[0],\ammR[1]) = \res[{\tokT[0],\tokT[1]}](\confG)$ and
$(\ammRi[0],\ammRi[1]) = \res[{\tokT[0],\tokT[1]}](\confGi)$.
Then, by~\Cref{lma:price-reserve-relation},
$\exchO[\confG]{\tokT[0],\tokT[1]} = \exchO[\confGi]{\tokT[0],\tokT[1]}$,
as the reserve ratios in $\confG$ and $\confGi$ are equal.
Thus, the amount of minted tokens $\valV$
received by the depositing user in \eqref{eq:usr-gain-solution} is fixed by
$(\valV[0],\valV[1])$, and the effect of a deposit inner layer does not depend
on its position along the Dagwood sandwich (\Cref{thm:n-tx-soln}).

Similarly to the construction of the swap inner layer,
$\pmvM$ may need to front-run transaction
\mbox{$\txT = \actAmmDeposit{\pmvA}{\valV[0]}{\tokT[0]}{\valV[1]}{\tokT[1]}$}
to enable it.
For
\mbox{$\confG = (\ammR[0]:\tokT[0],\ammR[1]:\tokT[1])\mid \cdots$},
we define the \keyterm{deposit front-run reserves} as:
\begin{equation*}
  \label{eq:canon-df-reserves}
  \resDeposit[0](\tokT[0],\tokT[1],\confG,\txT)
  = \left|\sqrt{\nicefrac{\valV[0]}{\valV[1]} \cdot \ammR[0] \cdot \ammR[1]}\:\right|
  \quad
  \resDeposit[1](\tokT[0],\tokT[1],\confG,\txT)
  = \left|\sqrt{\nicefrac{\valV[1]}{\valV[0]} \cdot \ammR[0] \cdot \ammR[1]}\:\right|
\end{equation*}
which satisfy
\(
  {\resDeposit[0](\tokT[0],\tokT[1],\confG,\txT)} \cdot {\valV[1]}
  =
  {\resDeposit[1](\tokT[0],\tokT[1],\confG,\txT)} \cdot {\valV[0]}
\),
as required by $\nrule{[Dep]}$.
Given a swap direction $d_{\pmvM}$,
we define the \keyterm{deposit front-run values} as:
\begin{equation*}
  \label{eq:canon-df-values}
  \begin{split}
    \frontDepVal[{{\swapDir[\pmvM]}}]{}{\tokT[0],\tokT[1],\confG}{\txT}{}
    & = \resDeposit[{\swapDir[\pmvM]}](\tokT[0],\tokT[1],\confG,\txT) - \ammR[{\swapDir[\pmvM]}]
    \\
    \frontDepVal[{{1-\swapDir[\pmvM]}}]{}{\tokT[0],\tokT[1],\confG}{\txT}{}
    & = \ammR[{1-\swapDir[\pmvM]}] - \resDeposit[{1-\swapDir[\pmvM]}](\tokT[0],\tokT[1],\confG,\txT)
  \end{split}
\end{equation*}

\noindent
If
\(
\frontDepVal[{{\swapDir[\pmvM]}}]{}{\tokT[0],\tokT[1],\confG}{\txT}{} > 0
\)
and
\(
\frontDepVal[{{1-\swapDir[\pmvM]}}]{}{\tokT[0],\tokT[1],\confG}{\txT}{} > 0
\)
holds for a swap direction $\swapDir[\pmvM]$,
then we define the \keyterm{deposit front-run transaction} as:
\begin{align*}
  \frontDepTx{\tokT[0],\tokT[1],\confG}{\txT}{}
  & = \actAmmSwap
    {\pmvM}{{\swapDir[\pmvM]}}
    {\frontDepVal[0]{}{\tokT[0],\tokT[1],\confG}{\txT}{}}{\tokT[0]}
    {\frontDepVal[1]{}{\tokT[0],\tokT[1],\confG}{\txT}{}}{\tokT[1]}
\end{align*}

\noindent
If the reserve ratio in the initial state
does not coincide with the ratio of deposited funds,
\ie $\valV[0]/\valV[1] \not= \ammR[0]/\ammR[1]$,
then the deposit inner layer is
$\frontDepTx{\tokT[0],\tokT[1],\confG}{\txT}{} \; \txT$,
and the balance required by $\pmvM$ is
\(
  \frontDepVal[{{\swapDir[\pmvM]}}]{}{\tokT[0],\tokT[1],\confG}{\txT}{}
  :\tokT[{{\swapDir[\pmvM]}}]
\).
Otherwise, the deposit inner layer is made just by $\txT$,
and the required balance is zero.

\mypar{Redeem inner layer}

By~\Cref{lma:constant-wealth-step}, redeem actions preserve the user's net worth,
\ie $\pmvA$'s gain is zero when firing
\mbox{$\txT = \actAmmRedeem{\pmvA}{\valV:(\tokT[0],\tokT[1])}$} in $\confG$:
\begin{equation*}
  \gain[{\confG}]{\pmvA}{\txT}
  \; = \;
  - \valV \cdot \exchO[\confG]{\tokT[0],\tokT[1]}
  +\valV[0] \cdot \exchO{\tokT[0]} + \valV[1] \cdot \exchO{\tokT[1]}
  \; = \;
  0
\end{equation*}

Unlike for the deposit inner layer,
redeem transactions \emph{increase} the users' gain when executed
in the game solution.
This is apparent when substituting in the above equation
$\exchO[\confG]{\tokT[0],\tokT[1]} = \exchOInit_{\confG}^{\textit{min}}(\tokT[0],\tokT[1])$
(as per \Cref{lma:price-minimum})
to express the user gain \emph{contribution} \eqref{eq:usr-gain-solution}
of the redeem action.
\[
  - \valV \cdot
  \exchOInit_{\confG}^{\textit{min}}(\tokT[0],\tokT[1])
  + \valV[0] \cdot \exchO{\tokT[0]}
  + \valV[1] \cdot \exchO{\tokT[1]}
   \geq 0
\]
Therefore, user redeem actions always \emph{reduce} $\pmvM$'s gain,
and so they are \emph{not} included in the solution.
Therefore, the redeem inner layer is always empty.

\begin{figure}[t]
  \scalebox{1}{\parbox{\textwidth}{%
      \begin{align}
        \nonumber
        & \walM{18:\tokT[0],50.5:\tokT[1]} \mid
          \confG =
          (100:\tokT[0],100:\tokT[1]) \mid \cdots
        \\
        \nonumber
        \xrightarrow{\frontSwapTx{\tokT[0],\tokT[1],\confG}{\txT}{}}
        \; & \walM{29.3:\tokT[0],37.8:\tokT[1]} \mid
          \confGi =
          (88.8:\tokT[0],112.7:\tokT[1]) \mid \cdots
        \\
        \nonumber
        \xrightarrow{\txT = \actAmmSwapL{\pmvA}{40}{\tokT[0]}{35}{\tokT[1]}} \;
        \; & \walM{29.3:\tokT[0],37.8:\tokT[1]} \mid
          \confGii =
          (128.8:\tokT[0],77.7:\tokT[1]) \mid \cdots
        \\
        \nonumber
        \xrightarrow{\frontDepTx{\tokT[0],\tokT[1],\confGii}{\txTi}{}}
        \; & \walM{71.4:\tokT[0],0:\tokT[1]} \mid
          \confGiii =
          (86.6:\tokT[0],115.5:\tokT[1]) \mid \cdots
        \\
        \nonumber
        \xrightarrow{\txTi = \actAmmDeposit{\pmvA}{30}{\tokT[0]}{40}{\tokT[1]}}
        \; & \walM{71.4:\tokT[0],0:\tokT[1]} \mid
          \confGiiii =
          (116.6:\tokT[0],155.5:\tokT[1]) \mid \cdots
        \\
        \nonumber
        \xrightarrow{\swapTx{\tokT[0],\tokT[1]}{\confGiiii}{}}
        \; & \walM{53.4:\tokT[0],20.8:\tokT[1]} \mid
          (134.6:\tokT[0],134.6:\tokT[1]) \mid \cdots
      \end{align}}}
  \caption{A Dagwood sandwich exploiting a user swap, deposit and redeem (dropped).}
  \label{fig:full-solution}
\end{figure}

\begin{example}
  \label{ex:full-game-soln}
  We now recast the full example in \Cref{sec:intro}
  as a MEV game, considering all three user transactions in the txpool:
  \begin{align*}
    \TxT =
    \{\:
    \actAmmSwapL{\pmvA}{40}{\tokT[0]}{35}{\tokT[1] }
    \:,\:
    \actAmmDeposit{\pmvA}{30}{\tokT[1]}{40}{\tokT[1]}
    \:,\:
    \actAmmRedeem{\pmvA}{10:(\tokT[0],\tokT[1])}
    \:\}
  \end{align*}
  The game solution is shown in \Cref{fig:full-solution}: note that we can
  reuse the swap inner layer from \Cref{ex:1txpool-game-sf}, since the initial
  state and user swap action are identical.
  Thus, we continue by constructing the deposit inner layer for user deposit $\txTi$ in state $\confGii = (128.8:\tokT[0],77.7:\tokT[1])$.
  Here, the deposit front-run reserves are:
  \begin{align*}
    \resDeposit[0](\tokT[0],\tokT[1],\confGii,\txTi)
    & = \abs{\sqrt{\nicefrac{30}{40} \cdot 128.8 \cdot 77.7}} = 86.6
    \\
    \resDeposit[1](\tokT[0],\tokT[1],\confGii,\txTi)
    & = \abs{\sqrt{\nicefrac{40}{30} \cdot 128.8 \cdot 77.7}} = 115.5
  \end{align*}
  Since the ratio of the deposit front-run reserves does not coincide with
  the reserve ratio in $\confGii$ ($86.6/115.5 \not= 128.8/77.7$), the deposit
  front-run by $\pmvM$ is necessary to enable the user deposit action.
  By choosing a swap direction $d_{\pmvM} = 1$,
  we obtain the positive deposit front-run values,
  which confirm the choice of the direction:
  \begin{equation*}
    \frontDepVal[0]{}{\tokT[0],\tokT[1],\confGii}{\txTi}{}
    = 128.8 - 86.6 \approx 42.2
    \quad
    \frontDepVal[1]{}{\tokT[0],\tokT[1],\confGii}{\txTi}{}
    = 115.5 - 77.7 \approx 37.8
  \end{equation*}
  Therefore, $\pmvM$'s deposit front-run transaction is:
  \[
    \frontDepTx{\tokT[0],\tokT[1],\confGii}{\txTi}{}
    \; = \;
    \actAmmSwap
    {\pmvM}
    {1}
    {42.2}{\tokT[0]}
    {37.8}{\tokT[1]}
  \]
  which requires a balance $\tokBal(\tokT[1]) \geq 37.8$.
  The deposit inner layer is obtained by prepending
  this transaction to $\pmvA$'s deposit.
  The redeem inner layer is empty, as shown before.
  By~\eqref{eq:price-minimization-tx}, the final layer
  to minimize the price of minted tokens is:
  \begin{equation*}
    \actAmmSwap
    {\pmvM}
    {1}
    {18.0}{\tokT[0]}
    {20.8}{\tokT[1]}
  \end{equation*}

  \noindent
  Summing up, the full Dagwood sandwich
  (see also \Cref{fig:full-solution}) is:
  \[
    \frontSwapTx{\tokT[0],\tokT[1],\confG}{\txT}{}
    \;\;
    \txT
    \;\;
    \frontDepTx{\tokT[0],\tokT[1],\confGii}{\txTi}{}
    \;\;
    \txTi
    \;\;
    \swapTx{\tokT[0],\tokT[1]}{\confGiiii}{}
  \]
  which requires an initial balance
  $\tokBal = \setenum{18.0:\tokT[0], 12.7 + 37.8:\tokT[1]}$ by $\pmvM$.
  By inspection of the Dagwood sandwich execution in \Cref{fig:full-solution},
  it can be seen that the miner has obtained a gain of approximately 5,700.
  \qed
\end{example}

\section{Related work} \label{sec:related}

Daian \emph{et al.} \cite{Daian19flash}
study the effect of transaction reordering
obtained through \emph{priority gas auctions}.
These are games between users who compete to include a bundle of transactions
in the next block, bidding on transaction fees to incentivize miners
to include their own bundle.
Notably, \cite{Daian19flash} finds empirical evidence of the fact that
the gain derived from transaction reorderings in decentralized exchanges (DEX)
exceeds the gain given by block rewards and transaction fees in Ethereum.
The same work also proposes a game model of priority gas auctions,
showing a Nash equilibrium for players to take turns bidding,
compatibly with behavior observed in the wild on Ethereum.
Our mining game differs from that in~\cite{Daian19flash},
since we assume a greedy adversary wanting to maximize its gain
at the expense of all the other users, exploiting arbitrages on AMMs.

Zhou \emph{et al.} \cite{Zhou21high} provide a theoretical framework to study the front-running on AMMs.
Two sandwich heuristics are studied:
the \emph{front-run \& back-run swap} sandwich, and the novel
\emph{front-run redeem \& back-run swap and deposit}.
The swap semantics used in \cite{Zhou21high} is simplified, compared to ours,
since no minimum amount of received tokens is enforced by the AMM,
users only perform swaps and hold no minted tokens 
(depositing and swapping agents are decoupled). 
Further, extractable value from arbitrage is considered separately.
In comparison, we emphasize that we propose a solution to attack all main user action types offered by leading AMMs, thereby extracting value from user submitted
swaps and deposits.
Our model also accurately model minted tokens: their value is dynamically affected
by miner and user swaps during the execution of the attack.
Thus, our game solution extracts the maximum value in a more concrete setting,
considering the victim transactions of both aforementioned attacks in \cite{Zhou21high}, and leaving no arbitrage opportunities unexploited.

More general ordering and injection of transactions by a rational agent is generally referred to as \emph{front-running}.
Eskandari \emph{et al.}~\cite{Eskandari19sok} provide a taxonomy
for various front-running attacks in blockchain applications and networks.
This taxonomy is expanded in \cite{Qin21quantifying} with liquidations,
sandwich attacks and arbitrage actions between DEX.

Some works investigate the problem of detecting front-running attacks on public blockchains. For example, in~\cite{Qin21quantifying}, Qin \emph{et al.} introduce front-running detection heuristics which are deployed to empirically study the presence of such attacks on public DeFi applications.
On the other hand, various fair ordering schemes have been proposed to mitigate front-running or exploitation of miner-extractable value.
However, simple commit-and-reveal schemes still leak information such as account balances.
Breidenbach \emph{et al.}~\cite{breidenbach2017hydra}
propose ``submarine commitments'',
which rely on k-anonymity to prevent any leaks from user commitments.
Baum \emph{et al.}~\cite{baump2dex} introduce a order-book based
DEX which delegates the matching of orders
to an out-sourced, off-chain multi-party computation committee.
Private user orders are not revealed to other participants, such that no front-running can occur in each privately-computed order matching round.
Ciampi \emph{et al.}~\cite{Ciampi21fairmm} introduce a market maker protocol
in which the strictly sequential trade history between an off-chain
market maker and traders are verifiable as a hash-chain.
Any subsequent reordering by the AMM is publicly provable:
collateral from the market maker incentivizes honest,
fair-ordering behaviour. Such work aims to provide alternative, front-running resistant designs with AMM-like functionality. In contrast, our work is intended to formalize the behaviour of current, mainstream AMMs in the presence of a rational adversary.

The DeFi community is developing tools to enable agents to extract value
from smart contracts: \eg, flashbots~\cite{flashbots} is a project
aiming to develop Ethereum implementations which support transaction bundles:
Rather than front-running individual transactions by adjusting their fees,
an agent can communicate a sequence or \emph{bundle} of transactions to the miner,
asking its inclusion in the next block.
Our game solutions could be implemented to solve for such bundles.


\section{Conclusions}
\label{sec:conclusions}

We have addressed the problem of adversaries
extracting value from AMMs interactions to the detriment of users.
We have constructed an \emph{optimal} strategy that adversaries can use
to extract value from AMMs, focussing on the widespread class of
constant-product AMMs.
Our results apply to any adversary with the power
to reorder, drop or insert transactions:
besides miners, this includes \emph{roll-up aggregators},
like \eg Optimism and StarkWare~\cite{optimism,starkware}.
Notably, our work shows that it is possible to extract value from
\emph{all} types of AMM transactions,
while previous works focus on extracting value from token swaps, only.

In practice, value is also extracted from AMMs by
colluding mining and non-mining agents:
for the Ethereum blockchain, agents can send
\emph{transaction bundles}~\cite{flashbots} to mining pools for block inclusion,
in return for a fee.
Our technique naturally applies to this setting,
where the actions of the miner are simply replaced by actions by the agent
submitting the transaction bundle.

We now discuss the simplifying assumptions (1-8) listed in~\Cref{sec:mining-game}.
(1) User balances do not limit the order in which transactions in
the \emph{txpool} can be executed.
In practice, in some cases it would be possible to perform a sequence of actions
by exploiting the funds received from previous actions.
We leave ordering constraints imposed by limited wallet
balances for future work.
(2) The adversary holds no minted tokens prior to executing
the game solution. Yet, the adversary can exploit an
(unbounded) initial balance of atomic tokens
to acquire minted tokens as part of the game solution by performing deposits.
The optimality of the Dagwood sandwich illustrates that this is not necessary.
(3) The size of the \emph{Dagwood sandwich} is unbounded.
In practice, a typical block of transactions
will include other transactions besides those
directed to AMMs, and so the adversary can find enough space for its
sandwiches by dropping non-AMM transactions.
During times of block-congestion, a constraint on the length of the Dagwood
sandwich will apply: we conjecture that solving such an optimization
is NP-hard, and regard this as an relevant question for future work.
(4) Prices of atomic tokens are fixed for the duration of the game:
the Dagwood sandwich will need to be recomputed should prices change.
(5) AMMs only hold atomic tokens. This is common in practice,
but we note that extending the mining game to
account for arbitrary nesting of minted tokens by AMM pairs
is an interesting direction of future research.
(6) No AMM swap fees and (7) no transaction fees are modelled:
the adversary's gain resulting from the Dagwood sandwich
is an upper bound to profitability as fees tend to zero.
Yet, fees affect this gain, so they should be taken into account
to construct an optimal strategy.
Furthermore, transaction fees may make it convenient for a miner
to include user redeem transactions in the sandwich,
while these are never exploited by our strategy.
(8) Besides fees, we abstract from the intervals that users can express
to constrain the amount of tokens received upon deposits and redeems
(we only model these constraints for swaps).
This is left for future work.

In this paper we have considered AMMs which implement the
constant-product swap invariant, like \eg Uniswap and SushiSwap.
A relevant research question is how to solve the MEV game under
different swap invariants, \eg those used by Curve Finance and SushiSwap.
Uniform frameworks which address this problem have been proposed
in~\cite{engel2021composing,ammTheory}
where swap invariants are abstracted
as functions subject to a given set of constraints.

\paragraph{Acknowledgements} 
Massimo Bartoletti is partially supported by 
Conv.\ Fondazione di Sardegna \& Atenei Sardi project
F75F21001220007 \emph{ASTRID}. James Hsin-yu Chiang
is supported by the PhD School of DTU Compute.

\bibliographystyle{splncs04}
\bibliography{main}

\iftoggle{arxiv}{%
\clearpage
\appendix
\section{Proofs}
\label{sec:proofs}



\wealthconservation*
\begin{proof}
  Follows from Lemma 3 (preservation of net worth) in \cite{ammTheory}.
\end{proof}

\constantwealthstep*
\begin{proof}
  Follows from Lemma 3 (preservation of net worth) in \cite{ammTheory}.
  \bartnote{CHECK!!}
\end{proof}

\pricereserverelation*
\begin{proof}
Let the \emph{projected minted token price} of $(\tokT[0],\tokT[1])$ at
reserve ratio $R>0$ in state $\confG$ be defined as:
\[
  \exchOInit^{\textit{R}}_{\confG}(\tokT[0],\tokT[1])
  = \frac{\ammRi[0]}{\supply{\confG}(\tokT[0],\tokT[1])}
    \cdot\exchOInit(\tokT[0])
  + \frac{\ammRi[1]}{\supply{\confG}(\tokT[0],\tokT[1])}
    \cdot\exchOInit(\tokT[1])
\]
where for the \emph{projected reserves} $(\ammRi[0],\ammRi[1])$, both
$\ammRi[0]\cdot \ammRi[1] = \ammR[0]\cdot \ammR[1]$ and $R = \ammRi[0]/\ammRi[1]$
hold. Thus, the projected minted token price can be rewritten entirely in terms
of token reserves and supply in $\confG$ and \emph{projected ratio} $R$:
\begin{equation}
  \label{proof:eq:price-reserve-relation-1}
  \exchOInit^{\textit{R}}_{\confG}(\tokT[0],\tokT[1])
  = \frac{\sqrt{\ammR[0]\cdot\ammR[1]\cdot R}}{\supply{\confG}(\tokT[0],\tokT[1])}
    \cdot\exchOInit(\tokT[0])
  + \frac{\sqrt{\ammR[0]\cdot\ammR[1]/R}}{\supply{\confG}(\tokT[0],\tokT[1])}
    \cdot\exchOInit(\tokT[1])
\end{equation}
We note that from \eqref{proof:eq:price-reserve-relation-1} and
\eqref{eq:price:minted} it follows that
\begin{equation}
  \label{proof:eq:price-reserve-relation-2}
\exchOInit^{\textit{R}}_{\confG}(\tokT[0],\tokT[1]) =
\exchO[\confG]{\tokT[0],\tokT[1]}
\quad \text{if} \quad
\begin{array}{l}
\res[{\tokT[0],\tokT[1]}](\confG) = (\ammR[0],\ammR[1])\\
R = \ammR[0]/\ammR[1]
\end{array}
\end{equation}
Alternatively, the \emph{projected minted token price} in a given
state $\confG$ can be interpreted as the minted token price in $\confGi$ of
execution $\walM{\tokBal}\mid\confG \rightarrow^{\txT}\walM{\_}\mid\confGi$
where $\txT$ is a miner swap action and the reserve ratio $\ammRi[0]/\ammRi[1]=R$
holds in $\confGi$ but not in $\confG$. By definition then, there exists
$\tokBal$ and swap $\txT$ for any reachable state $\confG$ and $R>0$,
such that $\walM{\tokBal}\mid\confG \rightarrow^{\txT}\walM{\_}\mid\confGi$
and $\exchOInit^{\textit{R}}_{\confG}(\tokT[0],\tokT[1]) = \exchO[\confGi]{\tokT[0],\tokT[1]}$ if  $R \not= \ammR[0]/\ammR[1]$.\\

We prove \Cref{lma:price-reserve-relation} by showing that for any $R$, the
\emph{projected minted token price} of a pair remains \emph{constant} for any execution.
Thus, if in two states $\confG,\confGi$ along an execution the AMM pair reserve
ratios both equal $R = \ammR[0]/\ammR[1] =\ammRi[0]/\ammRi[1]$, prices must
also be equal, thereby proving the lemma.
\begin{equation}
  \label{proof:eq:price-reserve-relation-3}
  \exchOInit^{\textit{R}}_{\confG}(\tokT[0],\tokT[1]) =
  \exchOInit^{\textit{R}}_{\confGi}(\tokT[0],\tokT[1]) =
  \exchO[\confG]{\tokT[0],\tokT[1]} =
  \exchO[\confGi]{\tokT[0],\tokT[1]}
\end{equation}
\\
We prove that the \emph{projected minted token price} remains constant for any
execution by induction.

\mypar{Base case: empty} For an empty step, the projected minted token price
remains constant (trivially).

\mypar{Induction step: deposit/redeem} For a deposit or redeem execution
\mbox{$\confG[n] \rightarrow^{\txT} \confG[n+1]$} the
following must hold for $c>0$ by definition of $\nrule{[Dep]}$ and $\nrule{[Rdm]}$
\begin{equation*}
  (c\cdot\ammR[0]^{n},c\cdot\ammR[1]^{n}) = (\ammR[0]^{n+1},\ammR[1]^{n+1})
  \qquad
  c\cdot\supply{\confG[n]}(\tokT[0],\tokT[1]) = \supply{\confG[n+1]}(\tokT[0],\tokT[1])
\end{equation*}
Thus, we can write the \emph{projected minted token price} in $\confG[n+1]$ in terms
of reserves and token supply in $\confG[n]$, such that the equality is apparent.
\begin{align*}
  \exchOInit^{\textit{R}}_{\confG[n+1]}(\tokT[0],\tokT[1])
  &= \frac{\sqrt{c^2\cdot\ammR[0]^{n}\cdot\ammR[1]^{n}\cdot R}}{c\cdot\supply{\confG[n]}(\tokT[0],\tokT[1])}
    \cdot\exchOInit(\tokT[0])
  + \frac{\sqrt{c^2\cdot\ammR[0]^{n}\cdot\ammR[1]^{n}/R}}{c\cdot\supply{\confG[n]}(\tokT[0],\tokT[1])}
    \cdot\exchOInit(\tokT[1]) \\
  &= \frac{\sqrt{\ammR[0]^{n}\cdot\ammR[1]^{n}\cdot R}}{\supply{\confG[n]}(\tokT[0],\tokT[1])}
  \cdot\exchOInit(\tokT[0])
  + \frac{\sqrt{\ammR[0]^{n}\cdot\ammR[1]^{n}/R}}{\supply{\confG[n]}(\tokT[0],\tokT[1])}
  \cdot\exchOInit(\tokT[1]) = \exchOInit^{\textit{R}}_{\confG[n]}(\tokT[0],\tokT[1])
\end{align*}
\mypar{Induction step: swap} For a swap execution
\mbox{$\confG[n] \rightarrow^{\txT} \confG[n+1]$} both the supply of minted tokens
and the reserve product is maintained by definition of $\nrule{Swap}$
\[
  \ammR[0]^{n}\cdot\ammR[1]^{n} = \ammR[0]^{n+1}\cdot\ammR[1]^{n+1}
  \quad
  \supply{\confG[n]}(\tokT[0],\tokT[1]) = \supply{\confG[n+1]}(\tokT[0],\tokT[1])
\]
Again, we can express the \emph{projected minted token price} in $\confG[n+1]$ in terms
of reserves and token supply in $\confG[n]$ to illustrate the equality.
\[
  \exchOInit^{\textit{R}}_{\confG[n+1]}(\tokT[0],\tokT[1])
  = \frac{\sqrt{\ammR[0]^{n}\cdot\ammR[1]^{n}\cdot R}}{\supply{\confG[n]}(\tokT[0],\tokT[1])}
    \cdot\exchOInit(\tokT[0])
  + \frac{\sqrt{\ammR[0]^{n}\cdot\ammR[1]^{n}/R}}{\supply{\confG[n]}(\tokT[0],\tokT[1])}
    \cdot\exchOInit(\tokT[1])
  = \exchOInit^{\textit{R}}_{\confG[n]}(\tokT[0],\tokT[1])
\]
Thus, we have shown that the projected minted token price remains constant
for all executions. Therefore, \eqref{proof:eq:price-reserve-relation-3} holds,
proving the lemma.
\qed
\end{proof}

\priceminimum*
\begin{proof}{lma:price-minimum}
The proof reuses the definition of the \emph{projected minted token price}
\eqref{proof:eq:price-reserve-relation-1} defined in the proof of
\Cref{lma:price-reserve-relation}: there, we showed that the
\emph{projected minted token price} for any given reserve ratio $R>0$
remains constant for all executions.
Thus, by definition \eqref{proof:eq:price-reserve-relation-1}, the projected
minted token price in $\confG$ for all $R>0$ is the minted token price \emph{range}
which can be achieved by executing a swap in any reachable state $\confG$.

To find $\exchOInit^{min}_{\confG}(\tokT[0],\tokT[1])$, we first determine the
$R$ for which $\exchOInit^{\textit{R}}_{\confG}(\tokT[0],\tokT[1])$ is minimized
in any reachable state $\confG$.
\[
  \frac{\partial}{\partial R}
  \exchOInit^{\textit{R}}_{\confG}(\tokT[0],\tokT[1])
  = \frac{\sqrt{\ammR[0]^{n}\cdot\ammR[1]^{n}}}{2\cdot\supply{\confG[n]}(\tokT[0],\tokT[1])\cdot\sqrt{R}}
    \cdot\exchOInit(\tokT[0])
  - \frac{\sqrt{\ammR[0]^{n}\cdot\ammR[1]^{n}}}{2\cdot\supply{\confG[n]}(\tokT[0],\tokT[1])\cdot\sqrt{R}\cdot R}
    \cdot\exchOInit(\tokT[1])
\]
Setting the expression above as equal to zero and then solving for $R = R^{min}$
we obtain
\[
  R^{min} = \frac{\ammR[0]}{\ammR[1]} = \frac{\exchOInit(\tokT[1])}{\exchOInit(\tokT[0])}
\]
Further, we have determined the projected minted token price \emph{minimum} since the second derivative is positive
\[
  \frac{\partial^2}{\partial R^2}
  \exchOInit^{\textit{R}}_{\confG}(\tokT[0],\tokT[1])
  = -\frac{\sqrt{\ammR[0]^{n}\cdot\ammR[1]^{n}}}{4\cdot\supply{\confG[n]}(\tokT[0],\tokT[1])\cdot\sqrt{R}\cdot R}
  + \frac{3\cdot\sqrt{\ammR[0]^{n}\cdot\ammR[1]^{n}}}{4\cdot\supply{\confG[n]}(\tokT[0],\tokT[1])\cdot\sqrt{R}\cdot R^2}
\]
\[
  = -\frac{\sqrt{\ammR[0]^{n}\cdot\ammR[1]^{n}}}{4\cdot\supply{\confG[n]}(\tokT[0],\tokT[1])\cdot\sqrt{\frac{\exchOInit(\tokT[1])}{\exchOInit(\tokT[0])}}\cdot \frac{\exchOInit(\tokT[1])}{\exchOInit(\tokT[0])^2}}
  + \frac{3\cdot\sqrt{\ammR[0]^{n}\cdot\ammR[1]^{n}}}{4\cdot\supply{\confG[n]}(\tokT[0],\tokT[1])\cdot\sqrt{\frac{\exchOInit(\tokT[1])}{\exchOInit(\tokT[0])}}\cdot \frac{\exchOInit(\tokT[1])}{\exchOInit(\tokT[0])^2}}
\]
\[
  =  \frac{2\cdot\sqrt{\ammR[0]^{n}\cdot\ammR[1]^{n}}}{4\cdot\supply{\confG[n]}(\tokT[0],\tokT[1])\cdot\sqrt{\frac{\exchOInit(\tokT[1])}{\exchOInit(\tokT[0])}}\cdot \frac{\exchOInit(\tokT[1])}{\exchOInit(\tokT[0])^2}} > 0
\]
Thus, the function $\exchOInit^{\textit{min}}_{\confG}(\tokT[0],\tokT[1])$ is given
as
\[
  \exchOInit^{\textit{min}}_{\confG}(\tokT[0],\tokT[1])
  =\exchOInit^{\exchOInit(\tokT[1])/\exchOInit(\tokT[0])}_{\confG}(\tokT[0],\tokT[1])
\]
By definition of the \emph{project minted token price}, a swap exists such that
the projected price for reserve ratio $R$ is achieved in the resulting state if the reserve ratio in $\confG$ is not equal to $R$. Otherwise the reserve ratio must
equal $R$, and thus the empty step achieves the projected price trivially.
We have shown that $\exchOInit^{R}_{\confG}(\tokT[0],\tokT[1]) \geq \exchOInit^{\textit{min}}_{\confG}(\tokT[0],\tokT[1])$ for any $R>0$, thereby proving the lemma.
\qed
\end{proof}

\zerotxsolution*
\begin{proof} 
  \Cref{thm:0-tx-solution:greedy} states that the solution to $(\confG,\emptymset)$ can be greedily constructed from canonical swaps for each AMM pair in $\confG$, thereby minimizing the prices of all minted tokens and net worth of users whilst maximizing the gain for the miner.

  We prove the lemma by showing that the \emph{price minimization swap}
  \eqref{eq:price-minimization-tx} for a pair
  $(\tokT[0],\tokT[1])$ minimizes the respective minted token price. Since all
  AMM actions affect single pair reserves only, the miner can minimize the
  minted token price in any order, thereby proving the lemma.

  To prove that the \emph{price minimization swap} minimizes the minted token price of a pair,
  we show that it updates the pair reserve ratio to $\ammR[0]/\ammR[1] = \exchO{\tokT[1]} / \exchO{\tokT[0]}$, which, as shown in the proof of \Cref{lma:price-minimum}, minimizes the price for all executions.

  \mypar{Case: $d = 0$} We assume the canonical swap direction to be $d=0$.
  By definition of the canonical swap values
  at page~\pageref{eq:price-minimization-values}, we have:
  \[
    \swapVal[0]{0}{\tokT[0],\tokT[1]}{\confG}{}
    = \sqrt{\tfrac{\exchO{\tokT[1]}}{\exchO{\tokT[0]}} \cdot \ammR[0] \cdot \ammR[1]} - \ammR[0]
  \]
  \[
    \swapVal[1]{0}{\tokT[0],\tokT[1]}{\confG}{}
    = \frac
    {\ammR[1] \cdot \swapVal[0]{0}{\tokT[0],\tokT[1]}{\confG}{}}
    {\ammR[0] + \swapVal[0]{0}{\tokT[0],\tokT[1]}{\confG}{}}
    = \frac
    {\ammR[1] \cdot \sqrt{\tfrac{\exchO{\tokT[1]}}{\exchO{\tokT[0]}} \cdot \ammR[0] \cdot \ammR[1]} - \ammR[0]\cdot\ammR[1]}
    {\sqrt{\tfrac{\exchO{\tokT[1]}}{\exchO{\tokT[0]}} \cdot \ammR[0] \cdot \ammR[1]}}
  \]
  Further, the reserve product invariant must hold before and after the price
  minimization swap in direction $d=0$. We show that this holds:
  \[
    (\ammR[0]+\swapVal[0]{0}{\tokT[0],\tokT[1]}{\confG}{})
    \cdot
    (\ammR[1]-\swapVal[1]{0}{\tokT[0],\tokT[1]}{\confG}{})
    =
    \sqrt{\tfrac{\exchO{\tokT[1]}}{\exchO{\tokT[0]}} \cdot \ammR[0] \cdot \ammR[1]} \cdot \frac{\ammR[0]\cdot\ammR[1]}{\sqrt{\tfrac{\exchO{\tokT[1]}}{\exchO{\tokT[0]}} \cdot \ammR[0] \cdot \ammR[1]}} = \ammR[0]\cdot\ammR[1]
  \]
  Finally, we can show that the resulting reserve ratio following the price minimization
  swap is indeed $\exchO{\tokT[1]} / \exchO{\tokT[0]}$, thereby minimizing the
  minted token price (see proof of \Cref{lma:price-minimum}).
  \[
    \frac{\ammR[0]+\swapVal[0]{0}{\tokT[0],\tokT[1]}{\confG}{}}{\ammR[1]-\swapVal[1]{0}{\tokT[0],\tokT[1]}{\confG}{}} =
    \frac{\sqrt{\tfrac{\exchO{\tokT[1]}}{\exchO{\tokT[0]}} \cdot \ammR[0] \cdot \ammR[1]}}{\frac{\ammR[0]\cdot\ammR[1]}{\sqrt{\tfrac{\exchO{\tokT[1]}}{\exchO{\tokT[0]}} \cdot \ammR[0] \cdot \ammR[1]}}}
    = \frac{\tfrac{\exchO{\tokT[1]}}{\exchO{\tokT[0]}} \cdot \ammR[0] \cdot \ammR[1]}{\ammR[0]\cdot\ammR[1]} = \frac{\exchO{\tokT[1]}}{\exchO{\tokT[0]}}
  \]
  \mypar{Case: $d=1$} Follows similarly and is omitted for brevity.
  \qed
\end{proof}

\ntxsolution*
\begin{proof} 

We restate the user gain \eqref{eq:usr-gain-solution} from the execution of a
game solution following~\Cref{lma:price-minimum}.
\begin{equation*}
  \begin{array}{l}
  \gain[\walM{\tokBal}\mid\walA{\tokBal[\pmvA]}\mid\confG]{\pmvA}{\bcB}
  \\[4pt]
  = \sum_{\tokT \in \TokU[0]} \tokBali[\pmvA](\tokT) \cdot \exchO{\tokT} - \tokBal[\pmvA](\tokT) \cdot \exchO{\tokT}
      + \sum_{\tokT \in \TokU[1]}
      \tokBali[\pmvA](\tokT) \cdot \exchOInit_{\confG}^{\textit{min}}(\tokT) - \tokBal[\pmvA](\tokT) \cdot \exchOInit_{\confG}(\tokT)
  \end{array}
\end{equation*}
Here, the prices are either of atomic ($\exchOInit_{\confG}(\tokT)$),
or minted tokens ($\exchOInit_{\confG}(\tokT)$ and $\exchOInit_{\confG}^{\textit{min}}(\tokT)$), all determined in the initial state $\confG$. Thus, the exploitation
of individual user actions by the miner is decided on the action's effect
the user \emph{token balance} only.

We prove \Cref{thm:n-tx-soln} by showing that the "inner layer" for each
user action type are \emph{optimal} when constructed
in any order from the submitted user actions in $\TxT$.

\mypar{Swap-inner-layer} Firsty, we show that the \emph{swap front-run}
by the miner will always minimize the amount of tokens \emph{received} by the user.
Let $\txT=\actAmmSwap{\pmvA}{0}{\valV[0]}{\tokT[0]}{\valV[1]}{\tokT[1]}$
where $\swapDir[\pmvA] = 0$ and
\[
  \walM{\_}\mid\confG \xrightarrow{\frontSwapTx{\tokT[0]}{\tokT[0]}{\confG,\txT}} \walM{\_}\mid\confGi
  \xrightarrow{\txT} \walM{\_}\mid\confGii
\]
If the execution of user swap $\txT$ results in the \emph{minimal} received output amount
$\valV[1]$ for $\pmvA$, then for $\res[{\tokT[0],\tokT[1]}](\confG) = (\ammR[0],\ammR[1])$, $\res[{\tokT[0],\tokT[1]}](\confGi) = (\ammRi[0],\ammRi[1])$ and $\res[{\tokT[0],\tokT[1]}](\confGii) = (\ammRi[0]+\valV[0],\ammRi[1]-\valV[1])$ the reserve product
invariant must hold by definition of $\nrule{[Swap]}$.
\[
  (\ammRi[0]+\valV[0])\cdot(\ammRi[1]-\valV[1]) = \ammRi[0]\cdot\ammRi[1] =
  \ammR[0]\cdot\ammR[1]
\]
Solving for $\ammRi[0]$, we can rewrite as:
\begin{equation*}
\begin{array}{c}
  (\ammRi[0]+\valV[0])\cdot(\frac{\ammR[0]\cdot\ammR[1]}{\ammRi[0]}-\valV[1]) = \ammR[0]\cdot\ammR[1]
\\
\ammR[0]\cdot\ammR[1] - \valV[0]\cdot\ammRi[0] + \frac{\valV[0]\cdot\ammR[0]\cdot\ammR[1]}{\ammRi[0]} -\valV[0]\cdot\valV[1] = \ammR[0]\cdot\ammR[1]
\\
\valV[1]\cdot{\ammRi[0]}^{2} + \valV[0]\cdot\valV[1]\cdot\ammRi[0]-\valV[0]\cdot\ammR[0]\cdot\ammR[1] = 0
\end{array}
\end{equation*}
The determinant to the quadratic equation is
\[
D = \valV[0]^2\cdot\valV[1]^2 + 4 \cdot \valV[0]\cdot \valV[1] \cdot \ammR[0]\cdot\ammR[1]
\]
Thus we can solve for \emph{positive} reserves $\ammRi[0],\ammRi[1]$ in state
$\confGi$ expressed in terms of
the swap parameters ($\valV[0],\valV[1]$) and reserves $\ammR[0],\ammR[1]$ in
initial state $\confG$,
which coincide with definitions of the \emph{swap front-run reserves}
for $\swapDir[\pmvA] = 0$
(the case $\swapDir[\pmvA] = 1$ is omitted for brevity).
\[
\ammRi[0] = \frac{-\valV[0]\cdot\valV[1] + \sqrt{\valV[0]^2\cdot\valV[1]^2 + 4 \cdot \valV[0]\cdot \valV[1] \cdot \ammR[0]\cdot\ammR[1]} }{2\cdot\valV[1]} \qquad
\ammRi[1] = \frac{\ammR[0]\cdot\ammR[1]}{\ammRi[0]}
\]
If $(\ammR[0],\ammR[1])=(\ammRi[0],\ammRi[1])$, then clearly no swap front-run
is required. Otherwise, the direction of the swap front-run depends on the value
of $\ammRi[0],\ammRi[1]$. For $\ammRi[0] > \ammR[0]$ and $\ammRi[1] < \ammR[0]$,
the swap-front run direction $\swapDir[\pmvM] = 0$ is implied. For $\ammRi[0] > \ammR[0]$ and $\ammRi[1] < \ammR[0]$, $\swapDir[\pmvM] = 1$. The \emph{swap front-run values}
\eqref{eq:canon-sf-values} follow from the difference between initial
and \emph{swap front-run} reserves.

Since the swap front-run always enables the user swap such that
the the minimum output amount is returned, this implies that the \emph{effect}
on the user token balance when executing the solution \eqref{eq:usr-gain-solution}
is solely determined by user swap parameters $(\valV[0],\valV[1])$: it is not
affected by its position in the full solution, enabling the greedy construction
of the swap-inner-layer in \Cref{thm:n-tx-soln}.

The optimality of the swap-inner-layer can be easily shown: For our assumed
user swap direction $\swapDir[\pmvA] = 0$, if
$- \valV[0]\cdot  \exchO{\tokT[0]} + \valV[1]\cdot \exchO{\tokT[1]} < 0$ holds,
then the contribution to the user gain \eqref{eq:usr-gain-solution} must be
negative, and furthermore, since by definition of $\nrule{[Swap]}$,
$\valV[1]$ is the minimum amount the user can receive,
the swap-inner-layer must be optimal.

If $- \valV[0]\cdot  \exchO{\tokT[0]} + \valV[1]\cdot \exchO{\tokT[1]} \geq 0$,
then the swap-inner-layer will be $(0,\emptyseq)$, since there the user
swap can never reduce the user gain in any game solution. We omit the case
$\swapDir[\pmvA] = 1$ for brevity.

\mypar{Deposit-inner-layer}
The optimality of the deposit-inner-layer
follows from \Cref{sec:game-soln}.
\mypar{Redeem-inner-layer}
The optimality of the redeem-inner-layer $(0,\emptyseq)$ follows from
\Cref{sec:game-soln}.
\qed
\end{proof}

}
{}

\end{document}